\documentclass[journal]{IEEEtran}
\usepackage{amsmath,amsfonts,amssymb}
\usepackage{amsthm}                 
\usepackage{mathtools}
\usepackage{algorithmic}
\usepackage{algorithm}
\usepackage{array}
\usepackage[caption=false,font=normalsize,labelfont=sf,textfont=sf]{subfig}
\usepackage{textcomp}
\usepackage{stfloats}
\usepackage{url}
\usepackage{verbatim}
\usepackage{graphicx}
\usepackage{cite}
\usepackage{bm}
\usepackage{booktabs}
\usepackage{microtype}
\usepackage{xcolor}
\usepackage[normalem]{ulem}
\usepackage{float}

\definecolor{IEEEAddBlue}{RGB}{0,0,0}
\definecolor{IEEEDelRed}{RGB}{255,0,0}
\newcommand{\added}[1]{\textcolor{IEEEAddBlue}{#1}}

\newenvironment{addedtext}{\color{IEEEAddBlue}}{}

\newcommand{\R}{\mathbb{R}}

\DeclareMathOperator{\Span}{span}
\DeclareMathOperator{\tr}{tr}
\DeclareMathOperator{\rank}{rank}

\DeclareMathOperator{\blkdiag}{blkdiag}

\theoremstyle{plain}
\newtheorem{theorem}{Theorem}
\newtheorem{lemma}[theorem]{Lemma}
\newtheorem{proposition}[theorem]{Proposition}

\theoremstyle{definition}

\newtheorem{remark}{Remark}

\usepackage[hidelinks]{hyperref}
\usepackage[nameinlink,capitalize,noabbrev]{cleveref}
\begin{document}
\bstctlcite{IEEEexample:BSTcontrol}
\title{Wireless Linear Computation Broadcast}

\author{Shuo~Tan,~\IEEEmembership{Graduate~Student~Member,~IEEE,}
        and~Syed~A.~Jafar,~\IEEEmembership{Fellow,~IEEE}
\thanks{The authors are with the Department of Electrical Engineering and Computer Science, University of California at Irvine, Irvine, CA 92697 USA (e-mail: shuot8@uci.edu; syed@uci.edu).}}

\maketitle
\begin{abstract}
A linear computation broadcast (LCBC) problem comprises $K$ users (receivers) and a transmitter. The users wish to compute various (vector) linear functions of a common dataset, and possess in advance heterogeneous side information corresponding to various other linear computations. The goal for the broadcast transmitter, who knows the dataset and all desired and side-information functions, is to satisfy all demands as efficiently as possible. Prior work has explored the information-theoretic capacity of LCBC for zero-error finite-field computation over an ideal (noiseless) broadcast channel (BC). This paper develops a wireless LCBC (WLCBC) framework for the Gaussian MIMO BC with noisy receiver-side information under a transmit power constraint and arbitrary antenna configurations. 
In the WLCBC setting, a Gaussian source is linearly precoded for broadcast; each user requests a linear function of the source and forms an estimate by linearly combining its channel observation with its noisy side information.
Assuming perfect channel state information, we cast the centralized joint linear transceiver design as a weighted sum-MSE minimization problem and propose an efficient alternating optimization algorithm.
For a fixed precoder, the optimal decoders are the linear minimum mean-square error (LMMSE) estimators.
For fixed decoders, the precoder update reduces to a convex quadratically constrained quadratic program which leads to a semi-closed-form solution parameterized by a single dual variable. The resulting algorithm guarantees a monotonic decrease in the objective and convergence of the weighted sum-MSE (WSMSE) objective sequence.
Simulations demonstrate pronounced robustness gains over natural baselines obtained from prior works.
\end{abstract}


\section{Introduction}\label{sec:intro}

Artificial intelligence (AI) is reshaping communication systems toward learning-assisted and computation-centric data dissemination~\cite{letaief2019roadmap,zhu2020toward,saad2019vision,gunduz2022beyond}.
The coupling of communication and computation has been predominantly explored in uplink multiple-access channel (MAC) settings, where waveform superposition facilitates over-the-air function computation~\cite{nazer2007computation,nazer2011compute} and model aggregation~\cite{yang2020federated,seif2025collaborative}.

The downlink counterpart presents a distinct challenge: the transmitter must simultaneously serve heterogeneous functional demands across users, each with user-specific \emph{side information}.

An information-theoretic abstraction for downlink function delivery is \emph{Computation Broadcast} (CBC)~\cite{sun2019capacity}, which studies the dissemination of user-specific functions of a common source in the presence of receiver-dependent side information.
This abstraction aligns with distributed edge deployments---e.g., autonomous driving~\cite{zhou_you_huang_2026}, UAV swarms~\cite{wei2025integrated}, and VR systems~\cite{kong2022edge}---where devices extract demand-relevant transforms of shared observations while leveraging locally cached context or auxiliary measurements~\cite{shi2020communication}.
A conceptually analogous situation arises when receiver-side cached model components serve as side information, and the network delivers low-dimensional demand descriptors (e.g., via LoRA-style adaptation~\cite{hu2022lora}).

Existing progress on CBC, especially its linear specialization, \emph{Linear Computation Broadcast} (LCBC), where the demands and side information are linear functions, has mainly focused on the information-theoretic capacity of LCBC for zero-error finite-field computation over an ideal (noiseless) broadcast channel (BC). The first CBC scheme in~\cite{sun2019capacity} is restricted to the $2$-user case.
The $3$-user LCBC solution in~\cite{yao2024capacity} is derived from linear rank inequalities for three desired spaces and three corresponding side-information spaces. Beyond three users, a symmetric $K$-user LCBC construction is developed in~\cite{yao2024generic} utilizing asymptotic interference alignment~\cite{cadambe2008interference}. An achievable scheme for general $K$-user LCBC was proposed in~\cite{ma2025achievable} using linear-programming-based design principles inspired by the $3$-user characterization in~\cite{yao2024capacity}.

Consequently, translating information-theoretic LCBC schemes into practical wireless downlink designs remains challenging.
LCBC schemes thus far are primarily developed over finite fields and under a zero-error criterion, whereas practical applications often require real-valued (over $\mathbb{R}$)~\cite{bishop2006pattern} and approximate computation under distortion metrics such as the mean-square error (MSE).
Moreover, existing LCBC formulations are largely physical-layer-agnostic: they model the broadcast medium as a noiseless common information pipe, \added{i.e., the trivial-channel setting in which every user receives the same broadcast information~\cite{sun2019capacity,yao2024capacity,yao2024generic,ma2025achievable}. They therefore} do not jointly account for MIMO antenna configurations, channel noise, and transmit power constraints, which are important considerations in wireless communications~\cite{nazer2011compute,lim2011noisy,nazer2016expanding}.
\added{To the best of our knowledge, the corresponding information-theoretic LCBC problem with user-dependent non-trivial channels remains open. By explicitly accounting for a distinct MIMO channel at each user, WLCBC provides a significant stepping stone toward this broader setting.}

\begin{addedtext}
Conventional MIMO downlink designs, such as weighted minimum mean-square error (WMMSE) \cite{christensen2008weighted,shi2011iteratively}, optimize linear filters for independent data streams, treating unintended signals as interference. In our setting, however, users compute correlated linear functions of a shared source. Consequently, a transmitted spatial direction can simultaneously benefit multiple users, with its utility governed by their heterogeneous side information.

Information-theoretic broadcast schemes using Slepian--Wolf, index, or lattice coding \cite{tuncel2006slepian,lee2015index,natarajan2015lattice,natarajan2018lattice} also exploit receiver side information. Yet, they inherently target the exact recovery of discrete messages, fundamentally differing from our objective of distortion-minimizing continuous estimation.

Closer to our continuous-amplitude setting are Wyner--Ziv and hybrid source--channel coding schemes \cite{wyner1976rate,costa1983writing,kochman2009joint,wilson2010joint}. While these works establish asymptotic limits for reconstruction under side information, they do not translate to finite-dimensional linear transceiver designs. Specifically, they lack mechanisms to jointly manage MIMO channels, noisy side observations, and user-specific computation demands.

To bridge this gap, we propose the wireless LCBC (WLCBC) framework for Gaussian MIMO broadcast channels. We jointly optimize the linear precoder and decoders to minimize the weighted sum-MSE (WSMSE) under a sum-power constraint. Crucially, each user's noisy side information induces a posterior residual covariance \(\bm\Pi_i\) that carefully weights the source directions based on their remaining uncertainty. Consequently, while the resulting block-coordinate updates retain familiar LMMSE and QCQP structures, their core mathematical operators are fundamentally shaped by the geometric interplay of user demands and side information.
\end{addedtext}

The main contributions are as follows:
\begin{itemize}
    \item \textbf{Centralized linear Gaussian MIMO formulation with noisy side information \added{and user-dependent channels}:}
    We formulate LCBC with noisy side information over a $K$-user real Gaussian MIMO broadcast channel under a transmit power constraint and perfect CSI, \added{explicitly accounting for a distinct MIMO channel at each user, as well as} heterogeneous antenna configurations and user-dependent noise levels.
    Each receiver has imperfect side information modeled as a correlated linear observation corrupted by additive Gaussian noise.
    Using the MSE between each user's demand and its estimate as the distortion metric, we obtain a centralized WSMSE-minimization framework that jointly optimizes the transmit precoder and the per-user linear decoders, while accounting for noisy channels and side information.

    \item \textbf{\added{WLCBC-aware LMMSE decoding via residual subspace geometry:}}
    \begin{addedtext}
        For a fixed precoder, the optimal LMMSE decoder is intrinsically shaped by the posterior residual covariance $\bm\Pi_i$. Geometrically, rather than enforcing a rigid orthogonal projection, $\bm\Pi_i$ acts as a structure that weights the demanded LCBC directions by their remaining uncertainty---smoothly attenuating well-observed components while preserving unobserved ones. Computationally, coupling this residual geometry with a thin-SVD reduction yields an efficient Schur-complement update, completely circumventing the high-dimensional inversion of the full composite covariance matrix.
    \end{addedtext}

    \item \textbf{Convex QCQP precoder update with a single dual variable:}
    For fixed decoders, we show that the precoder design reduces to a convex \emph{quadratically constrained quadratic program (QCQP)}.
    By analyzing the KKT conditions, we obtain a semi-closed-form solution parameterized by a single optimal Lagrange multiplier.
    This reduces the high-dimensional precoder optimization to a one-dimensional root-finding problem, which can be solved efficiently using a safeguarded Newton--bisection search.

    \item \textbf{Alternating optimization with monotonic descent and KKT characterization:}
    Combining the decoder and precoder updates yields an alternating optimization procedure with parallel per-user decoder updates. The resulting WSMSE objective sequence is monotonically non-increasing and convergent; moreover, every accumulation point of the iterate sequence satisfies the KKT conditions of the joint WSMSE design problem.
\end{itemize}

The remainder of this paper is organized as follows.
Section~\ref{sec:model} presents the system model.
Section~\ref{sec:method} develops the alternating optimization algorithm, with convergence and optimality proofs given in Section~\ref{sec:derivations}.
Section~\ref{sec:numerical_results} presents numerical simulations, and Section~\ref{sec:conclusion} concludes the paper.

\noindent\textit{Notation:}
Bold lowercase and uppercase letters denote vectors and matrices, respectively.
$(\cdot)^\top$ denotes transpose, $\mathbf I_n$ denotes the $n\times n$ identity matrix, and $\|\cdot\|_F$ denotes the Frobenius norm.
$\mathcal N(\mathbf 0,\mathbf \Sigma)$ denotes a real Gaussian distribution with mean $\mathbf 0$ and covariance $\mathbf \Sigma$.
For a matrix $\mathbf A$, $\rank(\mathbf A)$, and $\tr(\mathbf A)$ denote its rank and trace, respectively. For matrices $\mathbf A_1,\dots,\mathbf A_n$ the notation $\blkdiag(\mathbf A_1,\dots,\mathbf A_n)$ represents their block-diagonal concatenation.
The Kronecker product is denoted by $\otimes$.
\begin{figure*}[htbp]
  \centering
  \includegraphics[width=\linewidth]{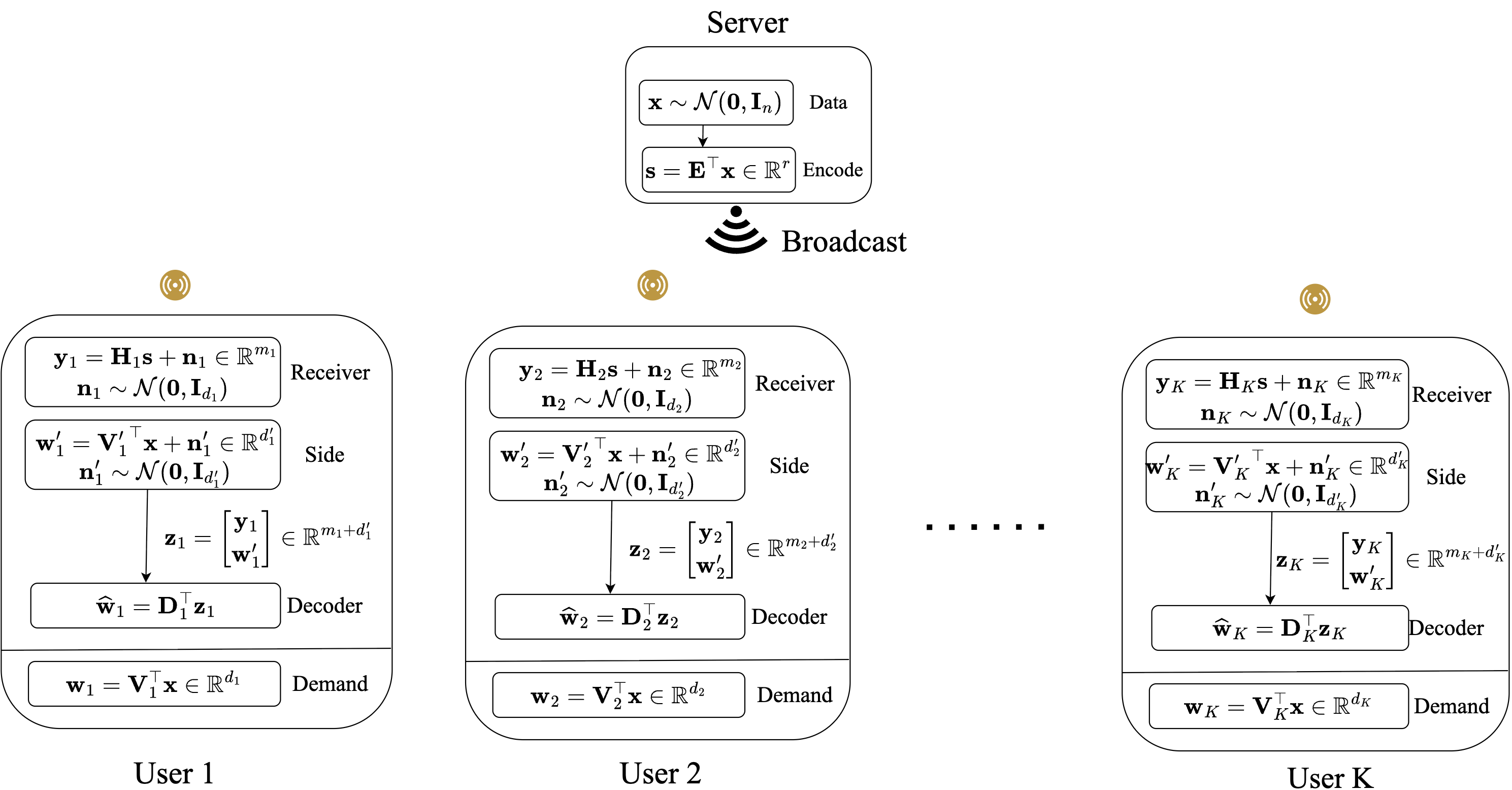}
  \caption{$K$-user WLCBC system model. The transmitter precodes $\mathbf{x}$ into $\mathbf{s}=\mathbf{E}^\top \mathbf{x}$; user $i$ combines $(\mathbf{y}_i,\mathbf{w}_i')$ to estimate $\mathbf{w}_i=\mathbf{V}_i^\top \mathbf{x}$.}
  \label{fig:system_model}
\end{figure*}
\section{System Model}\label{sec:model}
For a $K$-user WLCBC system, this section specifies (i) the Gaussian source
and users' linear computation demands, (ii) side information available at
each user, (iii) the $T$-use MIMO broadcast channel and its block-stacked
representation, and (iv) the linear transceiver architecture together with
the WSMSE design objective. 
\begin{addedtext}
Table~\ref{tab:dimension_guide} summarizes the key notation and dimensions used throughout the system model.
\end{addedtext}
\begin{table}[!ht]
\caption{\added{Dimension guide for the WLCBC system model.}}
\label{tab:dimension_guide}
\centering
\begin{addedtext}
\footnotesize
\setlength{\tabcolsep}{2pt}
\renewcommand{\arraystretch}{1.12}
\begin{tabular}{@{}p{0.38\columnwidth}p{0.56\columnwidth}@{}}
\toprule
\textbf{Notation} & \textbf{Definition and dimension} \\
\midrule
$T$, $M$, $N_i$
& Number of channel uses, transmit antennas, and receive antennas of user $i$. \\

$r=TM$, $m_i=TN_i$
& Block-extended transmit and receive dimensions. \\

$d_i$, $d_i'$
& Demand and side-information dimensions. \\

$\mathbf{x}\in\mathbb{R}^n$
& Common Gaussian source vector. \\

$\mathbf{V}_i\in\mathbb{R}^{n\times d_i}$; \newline
$\mathbf{w}_i=\mathbf{V}_i^\top\mathbf{x}\in\mathbb{R}^{d_i}$
& Demand matrix and target linear function. \\

$\mathbf{V}_i'\in\mathbb{R}^{n\times d_i'}$; \newline
$\mathbf{w}_i'={\mathbf{V}_i'}^\top\mathbf{x}+\mathbf{n}_i'\in\mathbb{R}^{d_i'}$

& Side-information matrix and noisy observation. \\
$\mathbf{E}\in\mathbb{R}^{n\times r}$; \newline
$\mathbf{s}=\mathbf{E}^\top\mathbf{x}\in\mathbb{R}^r$
& Centralized block precoder and transmit vector. \\

$\mathbf{H}_i\in\mathbb{R}^{m_i\times r}$; \newline
$\mathbf{y}_i=\mathbf{H}_i\mathbf{s}+\mathbf{n}_i\in\mathbb{R}^{m_i}$
& Block-stacked MIMO channel and received vector. \\
$\mathbf{z}_i\in\mathbb{R}^{m_i+d_i'}$
& Composite observation obtained by stacking $\mathbf{y}_i$ and $\mathbf{w}_i'$. \\

$\mathbf{D}_i\in\mathbb{R}^{(m_i+d_i')\times d_i}$; \newline
$\widehat{\mathbf{w}}_i=\mathbf{D}_i^\top\mathbf{z}_i$
& Linear decoder and target-function estimate. \\
\bottomrule
\end{tabular}
\end{addedtext}
\end{table}
\subsection{Source and computation demands}\label{subsec:source_demand}
We model the common source vector $\mathbf{x}\in\mathbb{R}^n$ as an
isotropic Gaussian random vector,
$\mathbf{x}\sim\mathcal{N}(\mathbf{0},\mathbf{I}_n)$.
Let $\mathcal{K}\triangleq\{1,\dots,K\}$ denote the user set.
Each user $i\in\mathcal{K}$ aims to compute a linear function of
$\mathbf{x}$:
\begin{equation}\label{eq:target_function}
  \mathbf{w}_i = \mathbf{V}_i^\top \mathbf{x} \in \R^{d_i},
\end{equation}
where $\mathbf{V}_i \in \R^{n \times d_i}$ is the demand matrix and $d_i$
is the demand dimension.

\subsection{MIMO broadcast channel}\label{subsec:channel}
Consider transmission from an $M$-antenna transmitter to an $N_i$-antenna user $i$ over a block of $T$ channel uses. 
\begin{addedtext}
To evaluate the distortion--communication tradeoff consistently across $T$, we assume a common underlying channel sequence $\{\mathbf H_i^{(\mathrm p)}[t]\}_{t\ge 1}$. A $T$-use model employs the first $T$ blocks, thereby strictly nesting the $T$-use system within the $(T+1)$-use system.
\end{addedtext}

\subsubsection{Per-use signal model}
At channel use $t\in\{1,\dots,T\}$, the transmitted signal is
$\mathbf{s}[t]\in\mathbb{R}^{M}$, and the received signal at user $i$ is
\begin{addedtext}
\begin{equation}\label{eq:per_use_channel}
  \mathbf{y}_i[t]
  =
  \mathbf{H}_i^{(\mathrm{p})}[t] \, \mathbf{s}[t]
  + \mathbf{n}_i[t],
  \qquad
  \mathbf{n}_i[t]\sim
  \mathcal{N}(\mathbf{0},\sigma_{c,i}^2\mathbf{I}_{N_i}),
\end{equation}
where $\mathbf{H}_i^{(\mathrm{p})}[t]\in\mathbb{R}^{N_i\times M}$ denotes
the per-use physical channel matrix and $\sigma_{c,i}^2$ is the channel-noise
variance.
\end{addedtext}
\subsection{Correlated side information}\label{subsec:sideinfo}
User $i$ is aided by correlated side information modeled as a noisy linear
observation of the source:
\begin{addedtext}
\begin{equation}\label{eq:side_info}
  \mathbf{w}_i' = {\mathbf{V}_i'}^\top \mathbf{x} + \mathbf{n}_i'
  \in \R^{d_i'}, \qquad
  \mathbf{n}_i' \sim
  \mathcal{N}(\mathbf{0},\sigma_{s,i}^2\mathbf{I}_{d_i'}).
\end{equation}
Here, $\mathbf{V}_i'\in\R^{n\times d_i'}$ characterizes the
side-information subspace, while $\sigma_{s,i}^2$ is the corresponding
noise variance.
\end{addedtext}
\subsubsection{Block stacking and extended dimensions}
Define the block-extended transmit dimension $r \triangleq TM$ and the
block-extended receive dimension of user $i$ as $m_i \triangleq TN_i$.
We stack the per-use signals as
\begin{equation}\label{eq:stacking_def}
\begin{aligned}
  \mathbf{s} &\triangleq
  \begin{bmatrix}\mathbf{s}[1]\\ \vdots\\ \mathbf{s}[T]\end{bmatrix}
  \in \mathbb{R}^{r}, &
  \mathbf{y}_i &\triangleq
  \begin{bmatrix}\mathbf{y}_i[1]\\ \vdots\\ \mathbf{y}_i[T]\end{bmatrix}
  \in \mathbb{R}^{m_i}.
\end{aligned}
\end{equation}
With this representation, \eqref{eq:per_use_channel} becomes
\begin{addedtext}
\begin{equation}\label{eq:channel_model}
\begin{aligned}
    \mathbf{y}_i &= \mathbf{H}_i\mathbf{s} + \mathbf{n}_i \in \mathbb{R}^{m_i}, \qquad \mathbf{n}_i \sim \mathcal{N}(\mathbf{0},\sigma_{c,i}^2\mathbf{I}_{m_i}), \\
    \mathbf{H}_i &\triangleq \mathrm{blkdiag}\!\bigl(\mathbf{H}_i^{(\mathrm{p})}[1],\dots,\mathbf{H}_i^{(\mathrm{p})}[T]\bigr) \in \mathbb{R}^{m_i\times r}.
\end{aligned}
\end{equation}
We write $\mathbf H_i^{[T]}$ when explicitly emphasizing the dependence on $T$, suppressing the superscript when $T$ is fixed.
\end{addedtext}

\subsubsection{Block-fading special case}
If $\mathbf{H}_i^{(\mathrm{p})}[t]\equiv \mathbf{H}_i^{(\mathrm{p})}$ for all
$t$, then
$\mathbf{H}_i = \mathbf{I}_T \otimes \mathbf{H}_i^{(\mathrm{p})}$.

\subsection{Linear precoding and power constraint}\label{subsec:precoding}
The transmitter adopts a linear precoding structure across the $T$ channel
uses via a precoding matrix $\mathbf{E}\in\R^{n\times r}$:
\begin{equation}\label{eq:precoder}
  \mathbf{s} = \mathbf{E}^\top \mathbf{x} \in \R^r.
\end{equation}
Equivalently, $\mathbf{E}=[\mathbf{E}[1],\dots,\mathbf{E}[T]]$ with
$\mathbf{E}[t]\in\R^{n\times M}$, so that
$\mathbf{s}[t]=\mathbf{E}[t]^\top\mathbf{x}$.

Let $P_{\mathrm{avg}}$ denote the average transmit-power budget per
antenna--time coordinate. 
\begin{equation}\label{eq:power_constraint}
    \frac{1}{TM}\sum_{t=1}^T \mathbb{E}\!\left[\|\mathbf{s}[t]\|^2\right]
    =
    \frac{1}{r}\mathbb{E}\!\left[\|\mathbf{s}\|^2\right]
    =
    \frac{1}{r}\|\mathbf{E}\|_F^2
    \le P_{\mathrm{avg}}.
\end{equation}
We normalize $P_{\mathrm{avg}}=1$, so the block power constraint is
$\|\mathbf{E}\|_F^2\le r$. \begin{addedtext}
All noise variables are mutually independent, and independent of the source $\mathbf{x}$.
\end{addedtext}

\subsection{Composite observation and linear decoding}\label{subsec:decoding}
User $i$ stacks its channel observation and side information as
\begin{equation}\label{eq:stacked_obs}
  \mathbf{z}_i \triangleq
  \begin{bmatrix}
  \mathbf{y}_i \\
  \mathbf{w}_i'
  \end{bmatrix}
  \in \R^{m_i+d_i'}.
\end{equation}
Combining \eqref{eq:channel_model}, \eqref{eq:precoder}, and
\eqref{eq:side_info} gives the linear Gaussian model
\begin{addedtext}
\begin{equation}\label{eq:equiv_linear_gaussian}
\mathbf{z}_i
=
\begin{bmatrix}
\mathbf{H}_i \mathbf{E}^\top \\
{\mathbf{V}_i'}^\top
\end{bmatrix}
\mathbf{x}
+
\widehat{\mathbf n}_i,
\;
\widehat{\mathbf n}_i\sim
\mathcal N\!\left(
\mathbf 0,
\operatorname{blkdiag}\!\left(
\sigma_{c,i}^2\mathbf I_{m_i},
\sigma_{s,i}^2\mathbf I_{d_i'}
\right)
\right).
\end{equation}
\end{addedtext}
Since $(\mathbf w_i,\mathbf z_i)$ are jointly Gaussian, the MMSE estimator
coincides with the LMMSE estimator~\cite{poor2013introduction}. Accordingly,
user $i$ employs a linear decoder
$\mathbf{D}_i\in\R^{(m_i+d_i')\times d_i}$ and forms
\begin{equation}\label{eq:estimator}
\widehat{\mathbf{w}}_i = \mathbf{D}_i^\top \mathbf{z}_i.
\end{equation}
The corresponding MSE is
\begin{equation}\label{eq:mse_def}
\mathcal{E}_i(\mathbf{E},\mathbf{D}_i)
\triangleq
\mathbb{E}\!\left[
\|\mathbf{w}_i-\mathbf{D}_i^\top\mathbf{z}_i\|^2
\right].
\end{equation}

\begin{addedtext}
Conformally partitioning the decoder according to \eqref{eq:stacked_obs} yields
\begin{equation}\label{eq:decoder_partition}
\mathbf D_i
=
\begin{bmatrix}
\mathbf D_{i,y}\\
\mathbf D_{i,w}
\end{bmatrix},
\qquad
\mathbf D_{i,y}\in\mathbb{R}^{m_i\times d_i},\quad
\mathbf D_{i,w}\in\mathbb{R}^{d_i'\times d_i}.
\end{equation}
Then the MSE exactly decomposes as
\begin{align}
\mathcal E_i(\mathbf E,\mathbf D_i)
={}&
\left\|
\mathbf V_i
-
\mathbf E\mathbf H_i^\top\mathbf D_{i,y}
-
\mathbf V_i'\mathbf D_{i,w}
\right\|_F^2
\nonumber\\
&+
\sigma_{c,i}^2\|\mathbf D_{i,y}\|_F^2
+
\sigma_{s,i}^2\|\mathbf D_{i,w}\|_F^2.
\label{eq:mse_lcbc_decomp}
\end{align}

The first term represents the squared residual of the noiseless algebraic recovery condition, while the remaining terms capture the noise amplification induced by the decoder. At finite SNR, enforcing exact algebraic alignment may be severely suboptimal if it dictates ill-conditioned decoder matrices. Consequently, \eqref{eq:mse_lcbc_decomp} reveals WLCBC as a physical-layer regularization of the classical LCBC problem, optimally balancing algebraic accuracy against noise robustness.
\end{addedtext}
\begin{addedtext}
\subsection{Centralized Information Assumptions and Local Execution}
\label{subsec:centralized_info}

We adopt a centralized-design and local-execution framework.
The transmitter is assumed to know
$
\{\mathbf H_i,\mathbf V_i,\mathbf V_i',
\omega_i,\sigma_{c,i}^2,\sigma_{s,i}^2\}_{i\in\mathcal K},
$
together with the transmit-power budget, but not the realizations of the
receiver noise terms $\mathbf n_i$ and $\mathbf n_i'$.

Using this information, the transmitter jointly computes the optimal precoder
$\mathbf E$ and the user-specific decoders
$\{\mathbf D_i\}_{i\in\mathcal K}$.
The transmitter uses $\mathbf E$ for transmission, and receiver $i$ forms
$
\widehat{\mathbf w}_i
=
(\mathbf D_i)^\top\mathbf z_i
$
from its local observation.
Thus, during execution, receiver $i$ requires only
$(\mathbf z_i,\mathbf D_i)$.
\end{addedtext}
\subsection{WSMSE Transceiver Design Objective}
\label{subsec:problem}
Our goal is to jointly design the linear precoder $\mathbf{E}$ and the per-user decoding matrices
$\{\mathbf{D}_i\}_{i\in\mathcal{K}}$ to minimize a weighted sum of users' MSEs under the transmit power constraint.
Define the global \emph{weighted sum-MSE} (WSMSE) objective as
\begin{equation}\label{eq:wsmse_objective}
F(\mathbf{E},\{\mathbf{D}_i\}) \triangleq \sum_{i\in\mathcal{K}} \omega_i\,\mathcal{E}_i(\mathbf{E},\mathbf{D}_i),
\end{equation}
where $\omega_i > 0$ denotes the positive weight for user $i$.
The joint transceiver design problem is formulated as
\begin{subequations}\label{eq:prob_joint}
\begin{align}
\min_{\mathbf{E},\,\{\mathbf{D}_i\}_{i\in\mathcal{K}}}\quad & F(\mathbf{E},\{\mathbf{D}_i\}) \label{eq:prob_joint_obj}\\
\text{s.t.}\quad & \|\mathbf{E}\|_F^2 \le r. \label{eq:prob_joint_pow}
\end{align}
\end{subequations}
\begin{remark}[Canonical normalization]
\label{rem:canonical}
Without loss of generality, we use the canonical normalization
$\mathbf x\sim\mathcal N(\mathbf 0,\mathbf I_n)$ and identity noise covariances.
For any full-rank source/noise covariances, invertible whitening maps reduce the model to this form; the induced linear transformations are absorbed into
$(\mathbf V_i,\mathbf V_i',\mathbf H_i,\mathbf D_i)$ and into the corresponding power constraint.
\end{remark}

The entire framework is illustrated in Fig.~\ref{fig:system_model}.

\section{Proposed Alternating Optimization}\label{sec:method}
We solve the WSMSE problem in \eqref{eq:prob_joint} via \emph{alternating optimization} (AO), i.e., exact block coordinate descent
over the precoder $\mathbf E$ and the decoder set $\{\mathbf D_i\}_{i\in\mathcal K}$.
Although \eqref{eq:prob_joint} is not jointly convex, it is convex in each block separately:
(i) for fixed $\mathbf E$, each user decoder update is a strictly convex quadratic problem admitting a unique LMMSE solution;
(ii) for fixed $\{\mathbf D_i\}$, the precoder update reduces to a convex QCQP with a KKT-characterized global minimizer.
We first perform a one-time preprocessing that removes noise-only side-information dimensions (without loss of optimality),
then present the decoder and precoder updates, and finally summarize the overall AO procedure in Algorithm~\ref{alg:wsmse_solver_compact}.
All proofs are deferred to Section~\ref{sec:derivations}.

Throughout this section, we adopt the stacked representation in Section~\ref{sec:model}:
$r=TM$, $m_i=TN_i$, and $\mathbf H_i\in\mathbb R^{m_i\times r}$ denotes the stacked channel matrix.

\subsection{Preprocessing: Rank reduction of noisy side information}\label{subsec:svd_reduction}
For user $i$, let $\widetilde{d}_i'\triangleq\rank(\mathbf V_i')$ and take the thin SVD
\begin{equation}\label{eq:svd_Vip}
\mathbf V_i'=\mathbf U_i'\bm\Sigma_i'{\mathbf T_i'}^\top,
\end{equation}
where $\mathbf U_i'\in\mathbb R^{n\times \widetilde{d}_i'}$ and $\mathbf T_i'\in\mathbb R^{d_i'\times \widetilde{d}_i'}$ have orthonormal columns and
$\bm\Sigma_i'\in\mathbb R^{\widetilde{d}_i'\times \widetilde{d}_i'}$ is diagonal.

Premultiplying \eqref{eq:side_info} by ${\mathbf T_i'}^\top$ yields the reduced side observation
\[
\widetilde{\mathbf w}_i' \triangleq {\mathbf T_i'}^\top \mathbf w_i'
= \bm\Sigma_i'\mathbf U_i'^\top \mathbf x + \widetilde{\mathbf n}_i',
\qquad
\widetilde{\mathbf n}_i' \sim \mathcal N(\mathbf 0,\mathbf I_{\widetilde{d}_i'}),
\]
since ${\mathbf T_i'}^\top$ preserves whiteness of $\mathbf n_i'$.

\begin{lemma}[Removal of noise-only side-information dimensions]\label{lem:wlog_side_reduction}
The orthogonal complement of $\Span(\mathbf T_i')$ contains pure noise independent of $\mathbf x$. It thus does not affect
(i) the optimal LMMSE decoder for any fixed $\mathbf E$, nor (ii) the optimal value of \eqref{eq:prob_joint}.
Hence, it is without loss of optimality to discard the noise-only dimensions and work with
$\widetilde{\mathbf z}_i \triangleq [\mathbf y_i^\top,(\widetilde{\mathbf w}_i')^\top]^\top$.
\end{lemma}

\begin{remark}[Notation reset after reduction]\label{rem:dim_reset}
Define the reduced side-information matrix $\widetilde{\mathbf V}_i' \triangleq \mathbf U_i'\bm\Sigma_i'\in\mathbb R^{n\times \widetilde{d}_i'}$ so that
$\widetilde{\mathbf w}_i' = (\widetilde{\mathbf V}_i')^\top\mathbf x + \widetilde{\mathbf n}_i'$.
After discarding the noise-only dimensions, we \emph{redefine}
\[
d_i' \leftarrow \widetilde{d}_i',\quad \mathbf w_i' \leftarrow \widetilde{\mathbf w}_i',\quad \mathbf z_i \leftarrow \widetilde{\mathbf z}_i,\quad
\mathbf V_i' \leftarrow \widetilde{\mathbf V}_i'.
\]
We still retain $(\mathbf U_i',\bm\Sigma_i')$ since they will be repeatedly used in the efficient decoder update.
\end{remark}

\subsection{Decoder update: Optimal decoder for fixed precoder}\label{subsec:opt_decoder}
Fix any feasible $\mathbf E$ and define $\mathbf R_i \triangleq \mathbf H_i\mathbf E^\top$.
Under the stacked MIMO model, $\mathbf E\in\mathbb R^{n\times r}$ with $r=TM$ can be partitioned as
$\mathbf E=[\mathbf E[1],\dots,\mathbf E[T]]$ with $\mathbf E[t]\in\mathbb R^{n\times M}$.
Then $\mathbf R_i$ can be formed blockwise without explicitly building the large stacked matrix $\mathbf H_i$:
\begin{equation}\label{eq:Ri_block_form}
  \mathbf R_i =
  \begin{bmatrix}
    \mathbf H_i^{(\mathrm{p})}[1]\,\mathbf E[1]^\top\\
    \vdots\\
    \mathbf H_i^{(\mathrm{p})}[T]\,\mathbf E[T]^\top
  \end{bmatrix}
  \in \mathbb{R}^{m_i \times n}.
\end{equation}

Using the reduced side information (Remark~\ref{rem:dim_reset}), the stacked observation admits the linear Gaussian form
\begin{equation}\label{eq:Ai_def_restate}
  \mathbf z_i = \mathbf A_i \mathbf x + \hat{\mathbf n}_i,
  \,
  \mathbf A_i \triangleq
  \begin{bmatrix} \mathbf R_i \\ (\mathbf V_i')^\top \end{bmatrix},
  \,
  \hat{\mathbf n}_i \sim \mathcal N(\mathbf 0,\mathbf I_{m_i+d_i'}).
\end{equation}

\begin{proposition}[LMMSE decoder and efficient evaluation]\label{prop:lmmse_decoder}
For fixed $\mathbf E$, the unique minimizer of $\mathcal E_i(\mathbf E,\mathbf D_i)$ is
\begin{equation}\label{eq:Gi_star_closed_form}
\mathbf{D}_i^\star
=
\big(\mathbf{A}_i\mathbf{A}_i^\top+\mathbf{I}_{m_i+d_i'}\big)^{-1}\mathbf{A}_i\mathbf{V}_i.
\end{equation}
Moreover, partition $\mathbf D_i^\star$ conformably with $\mathbf z_i=[\mathbf y_i^\top,(\mathbf w_i')^\top]^\top$ as
\[
\mathbf D_i^\star=\begin{bmatrix}\mathbf D_{i,y}^\star\\ \mathbf D_{i,w}^\star\end{bmatrix},
\quad
\mathbf D_{i,y}^\star\in\mathbb R^{m_i\times d_i},\ \mathbf D_{i,w}^\star\in\mathbb R^{d_i'\times d_i},
\]
and define
\begin{equation}\label{eq:system_matrices}
\begin{aligned}
\bm{\Phi}_{11} &\triangleq \mathbf I_{d_i'} + (\bm\Sigma_i')^2,\\
\bm{\Lambda}_i &\triangleq (\bm\Sigma_i')^2 \bm{\Phi}_{11}^{-1},\\
\bm{\Pi}_i &\triangleq \mathbf I_n-\mathbf U_i'\bm{\Lambda}_i{\mathbf U_i'}^\top \succ \mathbf 0.
\end{aligned}
\end{equation}
Then $\mathbf D_i^\star$ can be computed without inverting an $(m_i+d_i')\times(m_i+d_i')$ matrix via
\begin{align}
\mathbf{D}_{i,y}^\star
&=
\big(\mathbf{I}_{m_i} + \mathbf{R}_i \bm{\Pi}_i \mathbf{R}_i^\top\big)^{-1}
\, \mathbf{R}_i \bm{\Pi}_i \mathbf{V}_i, \label{eq:Giy_star_svd}\\[3pt]
\mathbf{D}_{i,w}^\star
&=
\bm{\Phi}_{11}^{-1} \bm{\Sigma}_i' \mathbf{U}_i'^\top
\Big(\mathbf{V}_i - \mathbf{R}_i^\top \mathbf{D}_{i,y}^\star\Big). \label{eq:Giw_star_svd}
\end{align}
\end{proposition}

\begin{addedtext}
The matrix \(\bm\Pi_i\) in \eqref{eq:system_matrices} acts as the posterior source covariance given the side information, i.e., \(\bm\Pi_i = \operatorname{Cov}(\mathbf x\mid \mathbf w_i')\). Consequently, the conditional covariances governing the channel observation are
\[
    \operatorname{Cov}(\mathbf y_i\mid \mathbf w_i') = \mathbf I_{m_i}+\mathbf R_i\bm\Pi_i\mathbf R_i^\top, \;
    \operatorname{Cov}(\mathbf y_i,\mathbf w_i\mid \mathbf w_i') = \mathbf R_i\bm\Pi_i\mathbf V_i.
\]
Thus, the channel-decoder block \(\mathbf D_{i,y}^\star\) in \eqref{eq:Giy_star_svd} is precisely the conditional LMMSE estimator of \(\mathbf w_i\) from \(\mathbf y_i\) given \(\mathbf w_i'\).
\end{addedtext}

\begin{remark}[\added{Residual covariance and LCBC geometry}]\label{rem:Pi_interp}
\begin{addedtext}
The posterior covariance \(\bm\Pi_i\) embeds the LCBC side-information geometry into the transceiver update. Its decomposition,
\[
    \bm\Pi_i = \mathbf P_{i,\perp} + \mathbf U_i' \big(\mathbf I_{d_i'} + (\bm\Sigma_i')^2\big)^{-1} {\mathbf U_i'}^\top,
    \; 
    \mathbf P_{i,\perp} \triangleq \mathbf I_n - \mathbf U_i'{\mathbf U_i'}^\top,
\]
reveals that \(\bm\Pi_i\) acts as a residual covariance rather than a rigid orthogonal projection. It preserves full uncertainty along the unobserved subspace \(\mathbf P_{i,\perp}\), while optimally attenuating each side-information direction according to its reliability. 

Consequently, the residual uncertainty of any demanded direction \(\mathbf v\), given by \(\mathbf v^\top\bm\Pi_i\mathbf v\), explicitly depends on its spatial alignment with the noisy side information. This structural weighting fundamentally distinguishes WLCBC from conventional MIMO designs that treat user demands as independent data streams.
\end{addedtext}
\end{remark}
\begin{addedtext}
\begin{remark}[Matrix-free LMMSE implementation and solve-domain selection]
\label{rem:schur_selection_short}
The quantities $\bm\Phi_{11}^{-1}$, $\bm\Lambda_i$, and $\mathbf S_i\triangleq{\mathbf U_i'}^\top\mathbf V_i$ depend only on the side-information model and can therefore be cached during preprocessing.

At each AO iteration, define $\mathbf G_i\triangleq\mathbf R_i\mathbf U_i'\in\mathbb R^{m_i\times d_i'}$. Then the coefficient matrix and right-hand side in \eqref{eq:Giy_star_svd} can be evaluated without explicitly forming the dense matrix $\bm{\Pi}_i$:
\begin{align}
\mathbf I_{m_i}+\mathbf R_i\bm{\Pi}_i\mathbf R_i^\top &= \mathbf I_{m_i}+\mathbf R_i\mathbf R_i^\top -\mathbf G_i\bm{\Lambda}_i\mathbf G_i^\top, \label{eq:matrix_free_coeff}\\
\mathbf R_i\bm{\Pi}_i\mathbf V_i &= \mathbf R_i\mathbf V_i -\mathbf G_i\bm{\Lambda}_i\mathbf S_i. \label{eq:matrix_free_rhs}
\end{align}
The resulting $m_i\times m_i$ coefficient matrix is symmetric positive definite (SPD) and can be factorized by Cholesky. After obtaining $\mathbf D_{i,y}^\star$, the side-information decoder block is recovered as $\mathbf D_{i,w}^\star = \bm\Phi_{11}^{-1}\bm\Sigma_i' (\mathbf S_i-\mathbf G_i^\top\mathbf D_{i,y}^\star)$, avoiding the explicit formation of $\mathbf V_i-\mathbf R_i^\top\mathbf D_{i,y}^\star$.

When $n<m_i$, the equivalent push-through identity
\[
(\mathbf I_{m_i}+\mathbf R_i\bm{\Pi}_i\mathbf R_i^\top)^{-1}\mathbf R_i\bm{\Pi}_i = \mathbf R_i(\bm{\Pi}_i^{-1}+\mathbf R_i^\top\mathbf R_i)^{-1}
\]
moves the main factorization to the source domain, where $\bm{\Pi}_i^{-1} = \mathbf I_n+\mathbf U_i'(\bm{\Sigma}_i')^2{\mathbf U_i'}^\top$. Hence, an implementation may strategically select the receive- or source-domain form according to the smaller effective system dimension. Detailed per-iteration operation counts are provided in Section~\ref{subsec:complexity}.
\end{remark}
\end{addedtext}
\subsection{Precoder update: Optimal solution of a convex QCQP for fixed decoders}\label{subsec:opt_precoder}
\begin{addedtext}
With the side-information geometry abstracted into the posterior covariances $\{\bm\Pi_i\}$, the precoder update transcends conventional stream-wise power allocation. Instead, it steers MIMO resources toward demand directions based jointly on their residual uncertainty, channel alignment, and decoder-induced weighting.
\end{addedtext}

Fix an arbitrary decoder set $\{\mathbf{D}_i\}_{i\in\mathcal{K}}$.
The corresponding precoder block of \eqref{eq:prob_joint} is
\begin{equation}\label{eq:precoder_subprob}
\min_{\mathbf{E}\in\R^{n\times r}} \ \sum_{i\in\mathcal{K}} \omega_i\,\mathcal{E}_i(\mathbf{E},\mathbf{D}_i)
\quad \text{s.t.}\quad \|\mathbf{E}\|_F^2\le r.
\end{equation}

Partition each fixed decoder conformably with
$\mathbf z_i=[\mathbf y_i^\top,(\mathbf w_i')^\top]^\top$ as
\[
\mathbf D_i=
\begin{bmatrix}
\mathbf D_{i,y}\\
\mathbf D_{i,w}
\end{bmatrix},
\qquad
\mathbf D_{i,y}\in\mathbb R^{m_i\times d_i},\quad
\mathbf D_{i,w}\in\mathbb R^{d_i'\times d_i}.
\]
Define (with $\mathbf V_i'$ denoting the reduced matrix after Remark~\ref{rem:dim_reset})
\begin{equation}\label{eq:Ci_Mi_def}
\mathbf{C}_i \triangleq \mathbf{V}_i^\top-\mathbf{D}_{i,w}^\top\mathbf{V}_i'^\top \in \R^{d_i\times n},
\;
\mathbf{B}_i \triangleq \mathbf{D}_{i,y}^\top\mathbf{H}_i \in \R^{d_i\times r}.
\end{equation}

The following lemma gives the quadratic reduction of the precoder subproblem; its proof is deferred to Section~\ref{sec:derivations}.

\begin{lemma}[Quadratic form of the precoder subproblem]\label{lem:mse_quadratic_E}
Fix $\{\mathbf{D}_i\}_{i\in\mathcal K}$ and define $\mathbf C_i,\mathbf B_i$ by \eqref{eq:Ci_Mi_def}.
Then the per-user MSE admits the decomposition
\begin{equation}\label{eq:mse_E_decomp}
    \mathcal{E}_i(\mathbf{E},\mathbf{D}_i)
    = \|\mathbf{C}_i-\mathbf{B}_i\mathbf{E}^\top\|_F^2
    + \|\mathbf{D}_{i,y}\|_F^2+\|\mathbf{D}_{i,w}\|_F^2.
\end{equation}
Consequently, after dropping the terms independent of $\mathbf E$, the precoder subproblem \eqref{eq:precoder_subprob} is equivalent to
\begin{equation}\label{eq:precoder_qcqp}
    \min_{\|\mathbf{E}\|_F^2\le r}\ \sum_{i\in\mathcal{K}}\omega_i\,\|\mathbf{C}_i-\mathbf{B}_i\mathbf{E}^\top\|_F^2.
\end{equation}
\end{lemma}

Let $d_{\mathrm{tot}} \triangleq \sum_{i\in\mathcal{K}} d_i$ and define the weighted concatenations
\begin{equation}\label{eq:BC_concat}
\begin{aligned}
  \mathbf{C} &\triangleq \big[\sqrt{\omega_1}\mathbf{C}_1^\top,\dots,\sqrt{\omega_K}\mathbf{C}_K^\top\big] \in \mathbb{R}^{n\times d_{\mathrm{tot}}}, \\
  \mathbf{B} &\triangleq \big[\sqrt{\omega_1}\mathbf{B}_1^\top,\dots,\sqrt{\omega_K}\mathbf{B}_K^\top\big] \in \mathbb{R}^{r\times d_{\mathrm{tot}}}.
\end{aligned}
\end{equation}
Then \eqref{eq:precoder_qcqp} can be written compactly as the convex QCQP
\begin{equation}\label{eq:precoder_ls}
\min_{\|\mathbf{E}\|_F^2\le r}\ \|\mathbf{C}-\mathbf{E}\mathbf{B}\|_F^2.
\end{equation}

\begin{proposition}[Optimal precoder structure]\label{prop:opt_precoder}
Given the fixed decoders $\{\mathbf{D}_i\}$, let $\mathbf Q\triangleq \mathbf B\mathbf B^\top\in\R^{r\times r}$
and $\mathbf P\triangleq \mathbf C\mathbf B^\top\in\R^{n\times r}$.
An optimal solution to the convex QCQP \eqref{eq:precoder_ls} is
\begin{equation}\label{eq:E_star_mu}
\mathbf{E}^\star(\mu)=\mathbf{P}\big(\mathbf{Q}+\mu\mathbf{I}_r\big)^{-1},
\end{equation}
where $\mu\ge 0$ satisfies complementary slackness: either $\mu=0$ if $\|\mathbf{E}^\star(0)\|_F^2 \le r$,
or $\mu>0$ such that $\|\mathbf{E}^\star(\mu)\|_F^2=r$.
(If $\mathbf Q$ is singular and $\mu=0$, we take $\mathbf E^\star(0)=\mathbf P\mathbf Q^\dagger$ as a minimum-norm minimizer.)
\end{proposition}

\begin{remark}[Implementation: safeguarded Newton for $\mu^\star$]\label{rem:implementation}
The power constraint \eqref{eq:power_constraint} is enforced via $\mu^\star$,
which can be computed efficiently as follows.
We avoid forming $(\mathbf Q+\mu\mathbf I_r)^{-1}$ explicitly.
If the unconstrained minimizer $\mathbf E^\star(0)=\mathbf P\mathbf Q^\dagger$ satisfies $\|\mathbf E^\star(0)\|_F^2\le r$,
then the constraint is inactive and $\mu^\star=0$. Otherwise, $\mu^\star>0$ is the unique root of $\phi(\mu)=r$, where
\[
\phi(\mu)\triangleq \|\mathbf E^\star(\mu)\|_F^2,\qquad
\mathbf E^\star(\mu)=\mathbf P(\mathbf Q+\mu\mathbf I_r)^{-1}.
\]
Since $\phi(\mu)$ is continuous and strictly decreasing on $\mu>0$ (cf.\ Proposition~\ref{prop:opt_precoder} and its proof),
we solve $\phi(\mu)=r$ via a safeguarded Newton (Newton--bisection) method.

A numerically stable way to evaluate $\phi(\mu)$ and $\phi'(\mu)$ is to diagonalize the \emph{positive semi-definite} (PSD) matrix $\mathbf{Q}$ once per AO iteration:
\[
\mathbf Q=\mathbf U\,\mathrm{diag}(\lambda_1,\dots,\lambda_r)\mathbf U^\top,\qquad \lambda_j\ge 0,
\]
and define $\widetilde{\mathbf P}\triangleq \mathbf P\mathbf U=[\tilde{\mathbf p}_1,\dots,\tilde{\mathbf p}_r]$.
Then
\[
\phi(\mu)=\sum_{j=1}^r \frac{\|\tilde{\mathbf p}_j\|_2^2}{(\lambda_j+\mu)^2},
\qquad
\phi'(\mu)= -2\sum_{j=1}^r \frac{\|\tilde{\mathbf p}_j\|_2^2}{(\lambda_j+\mu)^3}.
\]
Newton's update is
\[
\mu_{\rm N}\leftarrow \mu-\frac{\phi(\mu)-r}{\phi'(\mu)}
=
\mu + \frac{\phi(\mu)-r}{2\sum_{j=1}^r \frac{\|\tilde{\mathbf p}_j\|_2^2}{(\lambda_j+\mu)^3}}.
\]

To safeguard the iterations, maintain a bracket $[\mu_{\rm lo},\mu_{\rm hi}]$ such that
$\phi(\mu_{\rm lo})\ge r\ge \phi(\mu_{\rm hi})$.
A universal one-shot choice is $\mu_{\rm lo}=0$ and
\[
\mu_{\rm hi}=\frac{\|\mathbf P\|_F}{\sqrt r},
\]
which guarantees $\phi(\mu_{\rm hi})\le r$ (see the proof of Proposition~\ref{prop:opt_precoder}).
Starting from $\mu^{(0)}=(\mu_{\rm lo}+\mu_{\rm hi})/2$, iterate:
(i) compute $\mu_{\rm N}$ by Newton's update;
(ii) if $\mu_{\rm N}\notin(\mu_{\rm lo},\mu_{\rm hi})$ or is not finite, set $\mu_{\rm N}\leftarrow(\mu_{\rm lo}+\mu_{\rm hi})/2$;
(iii) update the bracket: if $\phi(\mu_{\rm N})>r$ set $\mu_{\rm lo}\leftarrow\mu_{\rm N}$, else set $\mu_{\rm hi}\leftarrow\mu_{\rm N}$.
Terminate when $|\phi(\mu)-r|/r\le \varepsilon_\mu$ for a prescribed tolerance $\varepsilon_\mu$.

Finally, compute $\mathbf E^\star(\mu^\star)$ by solving the SPD system
\[
(\mathbf Q+\mu^\star\mathbf I_r)\mathbf X=\mathbf P^\top,
\qquad
\mathbf E^\star(\mu^\star)=\mathbf X^\top,
\]
or equivalently via the diagonal formula
$
\mathbf E^\star(\mu^\star)=\widetilde{\mathbf P}\,\mathrm{diag}\big(\frac{1}{\lambda_1+\mu^\star},\dots,\frac{1}{\lambda_r+\mu^\star}\big)\mathbf U^\top.
$
\end{remark}
\begin{algorithm}[!htbp]
\caption{WLCBC-Aware Alternating Transceiver Design}
\label{alg:wsmse_solver_compact}
\begin{algorithmic}[1]
\REQUIRE
$\{\mathbf H_i,\mathbf V_i,\mathbf V_i',
\sigma_{c,i}^2,\sigma_{s,i}^2\}_{i\in\mathcal K}$,
weights $\{\omega_i\}_{i\in\mathcal K}$,
power budget, tolerance $\epsilon$, and maximum iteration count $t_{\max}$.

\ENSURE $\mathbf E^\star$, $\{\mathbf D_i^\star\}_{i\in\mathcal K}$.

\STATE \textbf{Preprocessing:} For each $i$, compute thin SVD $\mathbf V_i'=\mathbf U_i'\bm\Sigma_i'{\mathbf T_i'}^\top$ and store $(\mathbf U_i',\bm\Sigma_i')$.
\STATE Initialize a feasible $\mathbf E^{(0)}$ with $\|\mathbf E^{(0)}\|_F^2\le r$; set $t\leftarrow 0$.
\FOR{$i=1$ to $K$}
    \STATE Compute $\mathbf D_{i,y}^{(0)}$ and $\mathbf D_{i,w}^{(0)}$ via \eqref{eq:Giy_star_svd}--\eqref{eq:Giw_star_svd} using $\mathbf E^{(0)}$.
\ENDFOR
\STATE Evaluate $F^{(0)} \leftarrow F(\mathbf E^{(0)},\{\mathbf D_i^{(0)}\}_{i\in\mathcal K})$.

\REPEAT
\STATE $t\leftarrow t+1$.

\STATE \textbf{Decoder update (parallel over $i$):}
\FOR{$i=1$ to $K$}
    \STATE $\mathbf R_i \leftarrow \mathbf H_i(\mathbf E^{(t-1)})^\top$.
    \STATE Form $\bm{\Phi}_{11}$ and $\bm{\Lambda}_i$ via \eqref{eq:system_matrices}, and update $\mathbf D_{i,y}^{(t)}$, $\mathbf D_{i,w}^{(t)}$ via \eqref{eq:Giy_star_svd}--\eqref{eq:Giw_star_svd} using the matrix-free implementation in Remark~\ref{rem:schur_selection_short}.
\ENDFOR

\STATE \textbf{Precoder update:}
\STATE Form $\mathbf C,\mathbf B$ via \eqref{eq:BC_concat}; set $\mathbf P\leftarrow \mathbf C\mathbf B^\top$, $\mathbf Q\leftarrow \mathbf B\mathbf B^\top$.
\STATE Compute unconstrained solution $\mathbf E_0 \leftarrow \mathbf P\mathbf Q^\dagger$.
\IF{$\|\mathbf E_0\|_F^2 \le r$}
    \STATE $\mu^\star \leftarrow 0$; $\mathbf E^{(t)} \leftarrow \mathbf E_0$.
\ELSE
    \STATE Find $\mu^\star>0$ such that $\|\mathbf P(\mathbf Q+\mu^\star\mathbf I_r)^{-1}\|_F^2=r$ via the safeguarded Newton method in Remark~\ref{rem:implementation}.
    \STATE Solve $(\mathbf Q+\mu^\star\mathbf I_r)\mathbf X=\mathbf P^\top$; set $\mathbf E^{(t)} \leftarrow \mathbf X^\top$.
\ENDIF

\STATE Evaluate $F^{(t)} \leftarrow F(\mathbf E^{(t)},\{\mathbf D_i^{(t)}\}_{i\in\mathcal K})$.
\UNTIL{the stopping criterion \eqref{eq:stopping} holds or $t=t_{\max}$.}

\STATE \textbf{Final synchronization:}
\STATE Set $\mathbf E^\star\leftarrow \mathbf E^{(t)}$ and recompute all $\{\mathbf D_i^\star\}_{i\in\mathcal K}$ via \eqref{eq:Giy_star_svd}--\eqref{eq:Giw_star_svd} using $\mathbf E^\star$.
\STATE Construct the decoder in the original observation coordinates and
provide $\mathbf D_{i,\mathrm{dep}}^\star$ to user $i$.
\end{algorithmic}
\end{algorithm}
\subsection{Overall AO algorithm and stopping rule}\label{subsec:ao_algorithm}
Define $F^{(t)}\triangleq F(\mathbf E^{(t)},\{\mathbf D_i^{(t)}\})$.
We terminate the AO loop when the relative improvement satisfies
\begin{equation}\label{eq:stopping}
\frac{\big|F^{(t)}-F^{(t-1)}\big|}{F^{(t-1)}} \le \epsilon,
\end{equation}
or when a prescribed maximum iteration count is reached.

\begin{addedtext}

\subsection{Per-Iteration Complexity and Scalability}
\label{subsec:complexity}

We count dense arithmetic operations and exclude the one-time
side-information SVD preprocessing. Iteration-invariant quantities
$(\bm\Phi_{11}^{-1},\bm\Lambda_i,\mathbf S_i)$ are assumed to be cached.
Here, $d_i'$ denotes the reduced informative rank after
Remark~\ref{rem:dim_reset}, and
$d_{\mathrm{tot}}\triangleq\sum_{i=1}^K d_i$.

For user $i$, forming
$\mathbf R_i=\mathbf H_i\mathbf E^\top$ costs $\mathcal O(m_iMn)$
operations due to the block-diagonal structure of $\mathbf H_i$. Using
the matrix-free receive-domain implementation in
Remark~\ref{rem:schur_selection_short}, the decoder update requires
\begin{equation}
\begin{aligned}
\mathcal C_{\mathrm{dec},i}
=
\mathcal O\Big(&
m_iMn+m_i^2n+m_i^3
+m_i(n+m_i)d_i \\
&+
m_i(n+m_i+d_i)d_i'
\Big).
\end{aligned}
\label{eq:decoder_complexity}
\end{equation}
Thus, the additional per-iteration dependence on the reduced
side-information rank $d_i'$ is strictly linear. If the source-domain
form in Remark~\ref{rem:schur_selection_short} is used, the
factorization-related terms $\mathcal O(m_i^2n+m_i^3)$ are replaced by
$\mathcal O(m_in^2+n^3)$. The $K$ independent decoder updates can be
executed in parallel.

The precoder update consists of forming
$\{\mathbf C_i,\mathbf B_i\}_{i=1}^K$, accumulating $\mathbf P$ and
$\mathbf Q$, performing one spectral factorization of $\mathbf Q$, and
running an $N_\mu$-step safeguarded scalar dual search. Aggregating these
steps gives
\begin{equation}
\begin{aligned}
\mathcal C_{\mathrm{prec}}
=
\mathcal O\bigg(
&\sum_{i=1}^K(nd_id_i'+m_iMd_i)
+nrd_{\mathrm{tot}}+r^2d_{\mathrm{tot}} \\
&+r^3+nr^2+N_\mu r
\bigg).
\end{aligned}
\label{eq:precoder_complexity}
\end{equation}
Here, the terms $r^3+nr^2$ correspond to the eigendecomposition of
$\mathbf Q$ and the multiplication
$\widetilde{\mathbf P}=\mathbf P\mathbf U$, while $N_\mu r$ accounts for
the scalar Newton--bisection evaluations after the spectral quantities
are cached.

Consequently, one full AO iteration has serial work and parallel
decoder-update depth
\begin{align}
\mathcal C_{\mathrm{AO}}^{\mathrm{ser}}
&=
\sum_{i=1}^K \mathcal C_{\mathrm{dec},i}
+\mathcal C_{\mathrm{prec}},
\label{eq:overall_AO_complexity}\\
\mathcal C_{\mathrm{AO}}^{\mathrm{par}}
&=
\max_{i\in\mathcal K}\mathcal C_{\mathrm{dec},i}
+\mathcal C_{\mathrm{prec}}.
\label{eq:parallel_AO_complexity}
\end{align}

For homogeneous antenna and demand dimensions
($m_i=TN$, $r=TM$, $d_i=d$, and $d_i'=d'$), the dominant cubic dense
factorization burden under the receive-domain implementation is
$
\mathcal O(KT^3N^3+T^3M^3).
$
The first term arises from the $K$ receive-domain SPD factorizations of
size $TN$, and the second from the eigendecomposition of the
$TM\times TM$ matrix $\mathbf Q$. Since the scalar dual search costs only
$\mathcal O(N_\mu TM)$ once the spectral quantities are cached, it is not
the dominant operation. Other matrix-multiplication terms retained in
\eqref{eq:decoder_complexity} and \eqref{eq:precoder_complexity} may
dominate in regimes where $n$ or $d_{\mathrm{tot}}$ scales faster than
the block antenna dimensions.
\end{addedtext}

\begin{proposition}[Convergence and stationarity of AO]\label{prop:convergence}
Let $\{(\mathbf E^{(t)},\{\mathbf D_i^{(t)}\}_{i\in\mathcal K})\}_{t\ge 0}$ be the sequence generated by
Algorithm~\ref{alg:wsmse_solver_compact}, and define
\[
F^{(t)} \triangleq F(\mathbf E^{(t)},\{\mathbf D_i^{(t)}\}_{i\in\mathcal K}).
\]
Then the objective sequence $\{F^{(t)}\}_{t\ge 0}$ is monotonically non-increasing and convergent.
Moreover, the iterate sequence is bounded, and every accumulation point of
$\{(\mathbf E^{(t)},\{\mathbf D_i^{(t)}\}_{i\in\mathcal K})\}_{t\ge 0}$
is a coordinatewise minimum of \eqref{eq:prob_joint}.
Consequently, every accumulation point satisfies the KKT conditions of \eqref{eq:prob_joint}.
\end{proposition}

\section{Proofs}\label{sec:derivations}
\subsection{Proof of Lemma~\ref{lem:wlog_side_reduction}}
\begin{proof}
Let $\mathbf T_\perp$ be any orthonormal basis of the orthogonal complement of $\Span(\mathbf T)$, so that
$[\mathbf T,\mathbf T_\perp]$ is square orthonormal.
Premultiplying $\mathbf w'=\mathbf V'^\top \mathbf x+\mathbf n'$ by $[\mathbf T^\top,\mathbf T_\perp^\top]^\top$ yields
\[
\begin{bmatrix}\widetilde{\mathbf w}'\\ \mathbf w_\perp'\end{bmatrix}
=
\begin{bmatrix}\bm\Sigma \mathbf U^\top\mathbf x\\ \mathbf 0\end{bmatrix}
+
\begin{bmatrix}\widetilde{\mathbf n}'\\ \mathbf n_\perp'\end{bmatrix},
\]
where $\mathbf w_\perp'=\mathbf T_\perp^\top\mathbf n'$ is independent of $\mathbf x$ and contains no information about $\mathbf x$.
Hence, any linear estimator can set its coefficient on $\mathbf w_\perp'$ to zero without affecting MSE, and the optimal value is unchanged.
\end{proof}

\subsection{Proof of Proposition~\ref{prop:lmmse_decoder}}\label{subsec:proof_lmmse}
\begin{proof}
Fix user $i$ and omit the user index when clear. Under the reduced model, we have
\begin{equation}
    \mathbf{w}=\mathbf{V}^\top \mathbf{x},\quad
    \mathbf{z}=\mathbf{A}\mathbf{x}+\hat{\mathbf{n}},
\end{equation}
where $\mathbf{x}\sim\mathcal{N}(\mathbf{0},\mathbf{I}_n)$, $\hat{\mathbf{n}}\sim\mathcal{N}(\mathbf{0},\mathbf{I}_{m+d'})$, and $\mathbf{x}\perp \hat{\mathbf{n}}$.

The optimal linear MMSE estimator satisfies the orthogonality principle $\mathbb{E}[(\mathbf{w}-\mathbf{D}^\top\mathbf{z})\mathbf{z}^\top]=\mathbf{0}$, i.e.,
\begin{equation}\label{eq:orthogonality}
    \mathbb{E}[\mathbf{w}\mathbf{z}^\top]=\mathbf{D}^\top \mathbb{E}[\mathbf{z}\mathbf{z}^\top].
\end{equation}
We compute the cross-covariance and covariance matrices as follows:
\begin{equation}
\begin{split}
    \mathbb{E}[\mathbf{w}\mathbf{z}^\top]
    &= \mathbb{E}[\mathbf{V}^\top\mathbf{x}(\mathbf{A}\mathbf{x}+\hat{\mathbf{n}})^\top] \\
    &= \mathbf{V}^\top \mathbb{E}[\mathbf{x}\mathbf{x}^\top]\mathbf{A}^\top
    = \mathbf{V}^\top \mathbf{A}^\top,
\end{split}
\end{equation}
and
\begin{equation}
\begin{split}
    \mathbb{E}[\mathbf{z}\mathbf{z}^\top]
    &=\mathbf{A}\mathbb{E}[\mathbf{x}\mathbf{x}^\top]\mathbf{A}^\top+\mathbb{E}[\hat{\mathbf{n}}\hat{\mathbf{n}}^\top] \\
    &=\mathbf{A}\mathbf{A}^\top+\mathbf{I}_{m+d'}.
\end{split}
\end{equation}
Substituting these into \eqref{eq:orthogonality} yields
\begin{equation}
    \mathbf{V}^\top \mathbf{A}^\top = \mathbf{D}^\top(\mathbf{A}\mathbf{A}^\top+\mathbf{I}_{m+d'}).
\end{equation}
Since $\mathbf{A}\mathbf{A}^\top+\mathbf{I}_{m+d'}\succ \mathbf{0}$, the solution is unique and equals
\begin{equation}
    \mathbf{D}^\star = (\mathbf{A}\mathbf{A}^\top+\mathbf{I}_{m+d'})^{-1}\mathbf{A}\mathbf{V},
\end{equation}
which recovers \eqref{eq:Gi_star_closed_form}.

Partition the decoder $\mathbf{D}^\star$ and matrix $\mathbf{A}$ conformably into two blocks:
\begin{equation}\label{eq:partition_blocks}
  \mathbf{D}^\star = \begin{bmatrix} \mathbf{D}_y^\star \\ \mathbf{D}_w^\star \end{bmatrix},
  \qquad
  \mathbf{A} = \begin{bmatrix} \mathbf{R} \\ \bm{\Sigma}\mathbf{U}^\top \end{bmatrix}.
\end{equation}
Let $\bm\Phi\triangleq \mathbf{A}\mathbf{A}^\top+\mathbf{I}_{m+d'}$. The normal equation $\bm\Phi\mathbf{D}^\star=\mathbf{A}\mathbf{V}$ can be written in block form as
\begin{equation}
\begin{bmatrix}
\bm\Phi_{00} & \bm\Phi_{01}\\
\bm\Phi_{10} & \bm\Phi_{11}
\end{bmatrix}
\begin{bmatrix}
\mathbf{D}_y^\star\\
\mathbf{D}_w^\star
\end{bmatrix}
=
\begin{bmatrix}
\mathbf{R}\mathbf{V}\\
\bm\Sigma\mathbf{U}^\top\mathbf{V}
\end{bmatrix},
\end{equation}
where the block components are defined as
\begin{equation}
\begin{split}
    \bm\Phi_{00} &= \mathbf{I}_m+\mathbf{R}\mathbf{R}^\top, \quad \bm\Phi_{01} = \mathbf{R}\mathbf{U}\bm\Sigma, \\
    \bm\Phi_{10} &= \bm\Phi_{01}^\top, \quad\qquad \bm\Phi_{11} = \mathbf{I}_{d'}+\bm\Sigma^2.
\end{split}
\end{equation}
Eliminating $\mathbf{D}_w^\star$ via the Schur complement of $\bm\Phi_{11}$ gives:
\begin{equation}
  (\bm\Phi_{00} - \bm\Phi_{01}\bm\Phi_{11}^{-1}\bm\Phi_{10}) \mathbf{D}_y^\star
  = \mathbf{R}\mathbf{V} - \bm\Phi_{01}\bm\Phi_{11}^{-1}\bm\Sigma\mathbf{U}^\top\mathbf{V}.
\end{equation}
Using $\bm\Lambda=\bm\Sigma^2(\mathbf{I}_{d'}+\bm\Sigma^2)^{-1}$ and $\bm\Pi=\mathbf{I}-\mathbf{U}\bm\Lambda\mathbf{U}^\top$, this simplifies to
\begin{equation}
    (\mathbf{I}_m+\mathbf{R}\bm\Pi\mathbf{R}^\top)\mathbf{D}_y^\star = \mathbf{R}\bm\Pi\mathbf{V},
\end{equation}
which yields \eqref{eq:Giy_star_svd}. Finally, the second block equation gives
\begin{equation}
    \mathbf{D}_w^\star=\bm\Phi_{11}^{-1}\bm\Sigma\mathbf{U}^\top(\mathbf{V}-\mathbf{R}^\top\mathbf{D}_y^\star),
\end{equation}
which is \eqref{eq:Giw_star_svd}.
\end{proof}

\subsection{Proof of Lemma~\ref{lem:mse_quadratic_E}}\label{subsec:proof_lemma_E}
\begin{proof}
The estimate is $\widehat{\mathbf{w}}_i=\mathbf{D}_{i,y}^\top\mathbf{y}_i+\mathbf{D}_{i,w}^\top\mathbf{w}_i'$.
Substituting the signal model yields
\begin{equation}
\begin{split}
    \mathbf{w}_i-\widehat{\mathbf{w}}_i
    &= \Big(\mathbf{V}_i^\top-\mathbf{D}_{i,w}^\top\mathbf{V}_i'^\top-\mathbf{D}_{i,y}^\top\mathbf{H}_i\mathbf{E}^\top\Big)\mathbf{x} \\
    &\quad -\mathbf{D}_{i,y}^\top\mathbf{n}_i-\mathbf{D}_{i,w}^\top\mathbf{n}_i'.
\end{split}
\end{equation}
Using $\mathbf{x}\sim\mathcal{N}(\mathbf{0},\mathbf{I}_n)$, independent noises with identity covariances, and independence across
$\mathbf{x},\mathbf{n}_i,\mathbf{n}_i'$, the cross terms vanish, and the MSE equals \eqref{eq:mse_E_decomp}.
Dropping the terms constant in $\mathbf{E}$ yields \eqref{eq:precoder_qcqp}.
\end{proof}

\subsection{Proof of Proposition~\ref{prop:opt_precoder}}\label{subsec:proof_precoder}
\begin{proof}
Consider \eqref{eq:precoder_ls}: $\min_{\|\mathbf{E}\|_F^2\le r}\ \|\mathbf{C}-\mathbf{E}\mathbf{B}\|_F^2$.
Using $\mathbf{Q}=\mathbf{B}\mathbf{B}^\top$ and $\mathbf{P}=\mathbf{C}\mathbf{B}^\top$, the objective can be written as
\begin{equation}
    \|\mathbf{C}-\mathbf{E}\mathbf{B}\|_F^2
    =\tr(\mathbf{E}\mathbf{Q}\mathbf{E}^\top)-2\tr(\mathbf{P}^\top\mathbf{E})+\tr(\mathbf{C}\mathbf{C}^\top),
\end{equation}
where the last term is constant in $\mathbf{E}$.
The Lagrangian is
\begin{equation}
\begin{split}
    \mathcal{L}(\mathbf{E},\mu)
    &=\tr(\mathbf{E}\mathbf{Q}\mathbf{E}^\top)-2\tr(\mathbf{P}^\top\mathbf{E})
    +\mu(\|\mathbf{E}\|_F^2-r), \mu\ge 0.
\end{split}
\end{equation}
Stationarity gives
\begin{equation}
    \nabla_{\mathbf{E}}\mathcal{L}
    =2\mathbf{E}(\mathbf{Q}+\mu\mathbf{I}_r)-2\mathbf{P}=\mathbf{0},
\end{equation}
hence
\begin{equation}\label{eq:E_star_mu_proof_final}
    \mathbf{E}^\star(\mu)=\mathbf{P}(\mathbf{Q}+\mu\mathbf{I}_r)^{-1}.
\end{equation}
This proves \eqref{eq:E_star_mu} for any $\mu>0$ since $\mathbf{Q}+\mu\mathbf{I}_r\succ \mathbf{0}$.
For $\mu=0$, if $\mathbf Q$ is singular we may take $\mathbf E^\star(0)=\mathbf P\mathbf Q^\dagger$ as a minimum-norm minimizer.

It remains to characterize $\mu$.
If $\mathbf P=\mathbf 0$, then \eqref{eq:E_star_mu_proof_final} yields $\mathbf E^\star(\mu)=\mathbf 0$ for all $\mu\ge 0$ and the constraint is inactive.
Assume $\mathbf P\neq \mathbf 0$ hereafter.
For $\mu>0$, define
\begin{equation}
    \phi(\mu)\triangleq \|\mathbf{E}^\star(\mu)\|_F^2
    = \tr\!\big(\mathbf{P}(\mathbf{Q}+\mu\mathbf{I}_r)^{-2}\mathbf{P}^\top\big).
\end{equation}
Since $\mathbf{Q}+\mu\mathbf{I}_r\succ \mathbf{0}$, $\phi(\mu)$ is differentiable and
\begin{equation}
\begin{aligned}
\phi'(\mu) &= -2 \operatorname{tr} \left( \mathbf{P} (\mathbf{Q} + \mu \mathbf{I}_r)^{-3} \mathbf{P}^\top \right) \\
           &= -2 \left\| (\mathbf{Q} + \mu \mathbf{I}_r)^{-3/2} \mathbf{P}^\top \right\|_F^2 \\
           &< 0
\end{aligned}
\end{equation}
where strict negativity holds because $\mathbf P\neq \mathbf 0$ and $(\mathbf Q+\mu\mathbf I_r)^{-3/2}\succ\mathbf 0$.
Therefore, $\phi(\mu)$ is continuous and strictly decreasing on $\mu>0$, and satisfies $\phi(\mu)\to 0$ as $\mu\to\infty$.
By complementary slackness, if $\|\mathbf E^\star(0)\|_F^2\le r$ then $\mu^\star=0$; otherwise there exists a unique $\mu^\star>0$ such that $\phi(\mu^\star)=r$.

Moreover, since $\mathbf Q\succeq \mathbf 0$, we have $\|(\mathbf Q+\mu\mathbf I_r)^{-1}\|_2\le 1/\mu$ for $\mu>0$, hence
\begin{equation}
\begin{aligned}
\phi(\mu) &= \left\| \mathbf{P}(\mathbf{Q} + \mu\mathbf{I}_r)^{-1} \right\|_F^2 \\
           &\le \|\mathbf{P}\|_F^2 \cdot \left\| (\mathbf{Q} + \mu\mathbf{I}_r)^{-1} \right\|_2^2 \\
           &\le \frac{\|\mathbf{P}\|_F^2}{\mu^2}
\end{aligned}
\end{equation}
In particular, $\mu_{\rm hi}=\|\mathbf P\|_F/\sqrt r$ ensures $\phi(\mu_{\rm hi})\le r$ and can be used to bracket the unique root in a safeguarded Newton search (Remark~\ref{rem:implementation}).
\end{proof}

\subsection{Proof of Proposition~\ref{prop:convergence}}\label{subsec:proof_convergence}
\begin{proof}
At iteration $t$, the decoder update minimizes the objective for fixed $\mathbf E^{(t-1)}$, hence
\begin{equation}\label{eq:conv_pf_dec}
F(\mathbf E^{(t-1)},\{\mathbf D_i^{(t)}\}_{i\in\mathcal K})
\le
F(\mathbf E^{(t-1)},\{\mathbf D_i^{(t-1)}\}_{i\in\mathcal K}).
\end{equation}
Then, the precoder update minimizes the objective for fixed $\{\mathbf D_i^{(t)}\}_{i\in\mathcal K}$ under the constraint
$\|\mathbf E\|_F^2\le r$, hence
\begin{equation}\label{eq:conv_pf_enc}
F(\mathbf E^{(t)},\{\mathbf D_i^{(t)}\}_{i\in\mathcal K})
\le
F(\mathbf E^{(t-1)},\{\mathbf D_i^{(t)}\}_{i\in\mathcal K}).
\end{equation}
Combining \eqref{eq:conv_pf_dec} and \eqref{eq:conv_pf_enc} yields $F^{(t)} \le F^{(t-1)}$ for all $t\ge 1$,
so $\{F^{(t)}\}_{t\ge 0}$ is monotonically non-increasing.
Since each user MSE is nonnegative, we have $F^{(t)}\ge 0$ for all $t$, and therefore $\{F^{(t)}\}$ converges.

Next, define the precoder feasible set $\mathcal E \triangleq \{\mathbf E\in\mathbb R^{n\times r}:\|\mathbf E\|_F^2\le r\}$, which is compact.
Hence $\{\mathbf E^{(t)}\}_{t\ge 0}\subset \mathcal E$ is bounded.
For each user $i$, the decoder update is given in closed form by Proposition~\ref{prop:lmmse_decoder}:
\[
\mathbf D_i^\star(\mathbf E)
=
\big(\mathbf A_i(\mathbf E)\mathbf A_i(\mathbf E)^\top+\mathbf I_{m_i+d_i'}\big)^{-1}\mathbf A_i(\mathbf E)\mathbf V_i,
\]
where $\mathbf A_i(\mathbf E)$ depends affinely on $\mathbf E$.
Because $\mathbf A_i(\mathbf E)\mathbf A_i(\mathbf E)^\top+\mathbf I_{m_i+d_i'} \succeq \mathbf I_{m_i+d_i'} \succ \mathbf 0$ for all $\mathbf E \in \mathcal E$, the matrix inversion is well-defined and continuous. 
Thus, $\mathbf D_i^\star(\mathbf E)$ is a continuous function of $\mathbf E$.
Therefore, because $\mathcal E$ is compact, the image $\{\mathbf D_i^\star(\mathbf E):\mathbf E\in\mathcal E\}$ is compact, implying that $\{\mathbf D_i^{(t)}\}_{t\ge 0}$ is bounded for every $i\in\mathcal K$.
Hence the full iterate sequence $\{(\mathbf E^{(t)},\{\mathbf D_i^{(t)}\}_{i\in\mathcal K})\}_{t\ge 0}$ is bounded and admits at least one accumulation point.

Algorithm~\ref{alg:wsmse_solver_compact} is an exact two-block coordinate-descent method for the continuously differentiable problem \eqref{eq:prob_joint}:
for fixed $\mathbf E$, the decoder block is minimized exactly by Proposition~\ref{prop:lmmse_decoder};
for fixed $\{\mathbf D_i\}_{i\in\mathcal K}$, the precoder block is minimized exactly by Proposition~\ref{prop:opt_precoder}.
Specifically, the precoder minimizer is unique when the power constraint is active ($\mu^\star>0$). When inactive ($\mu^\star=0$), the minimum-norm selection $\mathbf P\mathbf Q^\dagger$ yields a well-defined single-valued update, ensuring that the block-minimization mapping is closed as required by \cite{tseng2001convergence}.
Therefore, by standard convergence results for exact two-block coordinate-descent methods \cite{grippo2000convergence,tseng2001convergence}, every accumulation point of the iterate sequence is a coordinatewise minimum of \eqref{eq:prob_joint}.

Finally, let $(\bar{\mathbf E},\{\bar{\mathbf D}_i\}_{i\in\mathcal K})$ be any accumulation point. 
Since it is coordinatewise minimal, each decoder block satisfies the unconstrained optimality condition $\nabla_{\mathbf D_i} F = \mathbf 0$, while $\bar{\mathbf E}$ satisfies the KKT conditions of the precoder subproblem $\min_{\mathbf E \in \mathcal E} F(\mathbf E, \{\bar{\mathbf D}_i\}_{i\in\mathcal K})$ (i.e., stationarity of the Lagrangian with complementary slackness for the constraint $\|\mathbf E\|_F^2 \le r$). 
Together, these are precisely the KKT conditions of the joint problem~\eqref{eq:prob_joint}.
\end{proof}

\section{Numerical Results}\label{sec:numerical_results}
This section evaluates the proposed WLCBC framework, also referred to as the over-the-air (OTA) setting, through Monte Carlo simulations under linear-Gaussian models.
We consider three experiments over two channel settings.
Experiment~1 uses an identity-channel AWGN broadcast link in a two-user comparison with the Sun--Jafar LCBC (SJ-LCBC)~\cite{sun2019capacity}, including an examination of the respective roles of precoder and decoder optimization.
Experiments~2 and~3 use i.i.d.\ real Gaussian MIMO broadcast channels.
Experiment~2 studies the channel-use-versus-distortion tradeoff across different antenna configurations, including heterogeneous per-user SNRs, the convergence behavior of AO, and sensitivity to random initialization, while Experiment~3 compares WLCBC with a matched direct-output MIMO-BC baseline to examine the impact of its encoder-subspace restriction. \begin{addedtext}
A key observation in Experiment~2 is that, at sufficiently high SNRs, a sharp
drop in MSE from the interference-limited regime to the noise-limited
regime as the number of channel uses increases may indicate, for each fixed
antenna configuration, the minimum number of channel uses required by the
considered linear WLCBC design to attain a prescribed distortion level.
This provides an empirical connection to the classical LCBC.
\end{addedtext}
\subsection{Common Monte Carlo setup and performance metric}\label{subsec:exp_setup}
For each user $i\in\mathcal K\triangleq\{1,\ldots,K\}$, the demand matrix $\mathbf V_i\in\mathbb R^{n\times d_i}$ and
the side-information matrix $\mathbf V_i'\in\mathbb R^{n\times d_i'}$ are generated as independent random orthonormal bases,
i.e., $\mathbf V_i^\top\mathbf V_i=\mathbf I_{d_i}$ and $\mathbf V_i'^{\!\top}\mathbf V_i'=\mathbf I_{d_i'}$,
corresponding to isotropically random $d_i$- and $d_i'$-dimensional subspaces of $\mathbb R^n$.

Whenever channel noise and side-information noise are isotropic Gaussian, we parameterize their variances by the channel and side-information SNRs (in dB), denoted by $\gamma_c$ and $\gamma_s$, respectively, via
\begin{equation}\label{eq:sigma_from_snr_exp}
\sigma_{\mathrm{ch}}^2 = 10^{-\gamma_c/10},\qquad
\sigma_{\mathrm{side},i}^2 = 10^{-\gamma_s/10}.
\end{equation}
Unless otherwise stated, we use uniform user weights $\omega_i=1/|\mathcal K|$.

The AO-based WLCBC optimization terminates when the iteration count reaches $t_{\max}=250$ or when the relative convergence tolerance falls below $\epsilon=10^{-10}$.
For the precoder subproblem, the Lagrange multiplier $\mu^\star$ is computed via the safeguarded Newton method described in Remark~\ref{rem:implementation}.
\begin{addedtext}
We report the MSE for user $i$ as
$
\mathcal{E}_i(\mathbf{E},\mathbf{D}_i)
\triangleq
\mathbb{E}\!\left[
\|\mathbf{w}_i-\mathbf{D}_i^\top\mathbf{z}_i\|^2
\right].
$
The corresponding weighted sum-MSE is
\begin{equation}\label{eq:rel_mse_total_clip_exp}
\mathcal{E}_{\mathrm{total}}
=\sum_{i \in \mathcal{K}} \omega_i\,
\mathcal{E}_i(\mathbf{E}, \mathbf{D}_i).
\end{equation}
This quantity is reported in
Figs.~\ref{fig:simple_2user}, \ref{fig:complex_2user},
\ref{fig:exp2_mimo_Tsweep}, \ref{fig:ao_convergence}, and
\ref{fig:wlcbc_same_task_mimo}, whereas
Fig.~\ref{fig:heterogeneous_user_snr} reports the individual user MSEs
$\{\mathcal E_i\}_{i\in\mathcal K}$.
\end{addedtext}
\subsection{Experiment 1: WLCBC versus SJ-LCBC baseline (two-user)}\label{subsec:exp1}
We specialize in $K=2$ and benchmark WLCBC against SJ-LCBC.
The purpose of this experiment is to compare robustness against a canonical LCBC reference in a noisy setting; it is not intended as a comparison against all possible noisy-broadcast baselines.

In this experiment, we use an identity-channel SISO AWGN broadcast link over
$T=r$ channel uses, i.e., $M=N_i=1$ and
$\mathbf H_i=\mathbf I_r$ for all $i$, so that $m_i=r$ and
\begin{equation}\label{eq:noise_models_exp}
\mathbf n_i \sim \mathcal N(\mathbf 0, \sigma_{\mathrm{ch}}^2 \mathbf I_{r}), \qquad
\mathbf n_i' \sim \mathcal N(\mathbf 0, \sigma_{\mathrm{side},i}^2 \mathbf I_{d_i'}).
\end{equation}
In the SNR sweeps below, we set $\gamma_c=\gamma_s$ and denote the common SNR by $\gamma$.

Let $c\in\{\mathrm{Base},\mathrm{OTA}\}$ denote the scheme index; the precoder is $\mathbf E_c\in\mathbb R^{n\times r}$ and
the user-$i$ decoder is $\mathbf D_{i,c}\in\mathbb R^{(r+d_i')\times d_i}$, so that
\[
\widehat{\mathbf w}_{i,c}
=
\mathbf D_{i,c}^\top
\begin{bmatrix}
\mathbf y_i\\
\mathbf w_i'
\end{bmatrix}.
\]
Both schemes are compared using $\mathcal E_{\mathrm{total}}$ in~\eqref{eq:rel_mse_total_clip_exp} under the same constraint $\|\mathbf E_c\|_F^2\le r$.

\textbf{Baseline (SJ-LCBC).}
SJ-LCBC is an algebraic alignment-based construction designed for noiseless side information (and noiseless channels)~\cite{sun2019capacity}.
In each trial, after generating $(\mathbf V_i,\mathbf V_i')$, we fix the broadcast dimension $r$ for \emph{both} schemes
according to the noiseless subspace-decomposition prescription in~\cite{sun2019capacity}
(which yields $r=1$ for the toy instance below and $r=4$ for the heterogeneous instance below),
construct the corresponding SJ precoder/decoder, and scale the precoder to satisfy $\|\mathbf E_{\mathrm{Base}}\|_F^2\le r$.
The resulting SJ scheme is then applied directly to the noisy observations.

\textbf{Proposed WLCBC (OTA).}
We use over-the-air (OTA) to refer to the WLCBC setting, which jointly optimizes the precoder and decoders by alternating minimization of the MSE objective under channel noise and noisy side information.

\emph{(a) Toy instance:}
We first consider $(K,n,r,d_i,d_i')=(2,2,1,1,1)$, which appears in~\cite{sun2019capacity} as a canonical two-user example.
Let $\mathbf v_i,\mathbf v_i'\in\mathbb R^2$ be the columns of $\mathbf V_i,\mathbf V_i'$, define
$\mathbf G\triangleq[\mathbf v_1'\ \mathbf v_2']$, and let $\mathbf e_1=[1,0]^\top$, $\mathbf e_2=[0,1]^\top$.
The (unnormalized) SJ precoder is
\begin{equation}\label{eq:exp1_Eraw}
\mathbf E_{\mathrm{raw}}
=
[\mathbf v_1' \ \mathbf 0]\mathbf G^{-1}\mathbf v_2
+
[\mathbf 0 \ \mathbf v_2']\mathbf G^{-1}\mathbf v_1,
\end{equation}
which we normalize as $\mathbf E_{\mathrm{Base}}=\sqrt r\,\mathbf E_{\mathrm{raw}}/\|\mathbf E_{\mathrm{raw}}\|_F$.
For user $i\in\{1,2\}$ (with $j=3-i$), the baseline decoder is a two-tap linear combiner of the received symbol
and the one-dimensional side-information observation:
\begin{equation}\label{eq:exp1_Dbase}
\mathbf D_{i,\mathrm{Base}}
=
\begin{bmatrix}
\mathbf e_i^\top \mathbf G^{-1}(\mathbf v_i-\mathbf v_j)\\[2pt]
\|\mathbf E_{\mathrm{Base}}\|_F
\end{bmatrix}
\in \mathbb R^{2\times 1}.
\end{equation}
Fig.~\ref{fig:simple_2user} plots $\mathcal E_{\mathrm{total}}$ versus the common SNR $\gamma$.
As $\gamma$ increases, the gap shrinks, and the two schemes nearly coincide at $\gamma=20$~dB.
This behavior is consistent with the fact that SJ-LCBC is tailored to the noiseless/alignment regime, whereas WLCBC is substantially more robust in the noisy regime and approaches the same high-SNR limit as noise vanishes.
\begin{figure}[t]
  \centering
  \includegraphics[width=\columnwidth]{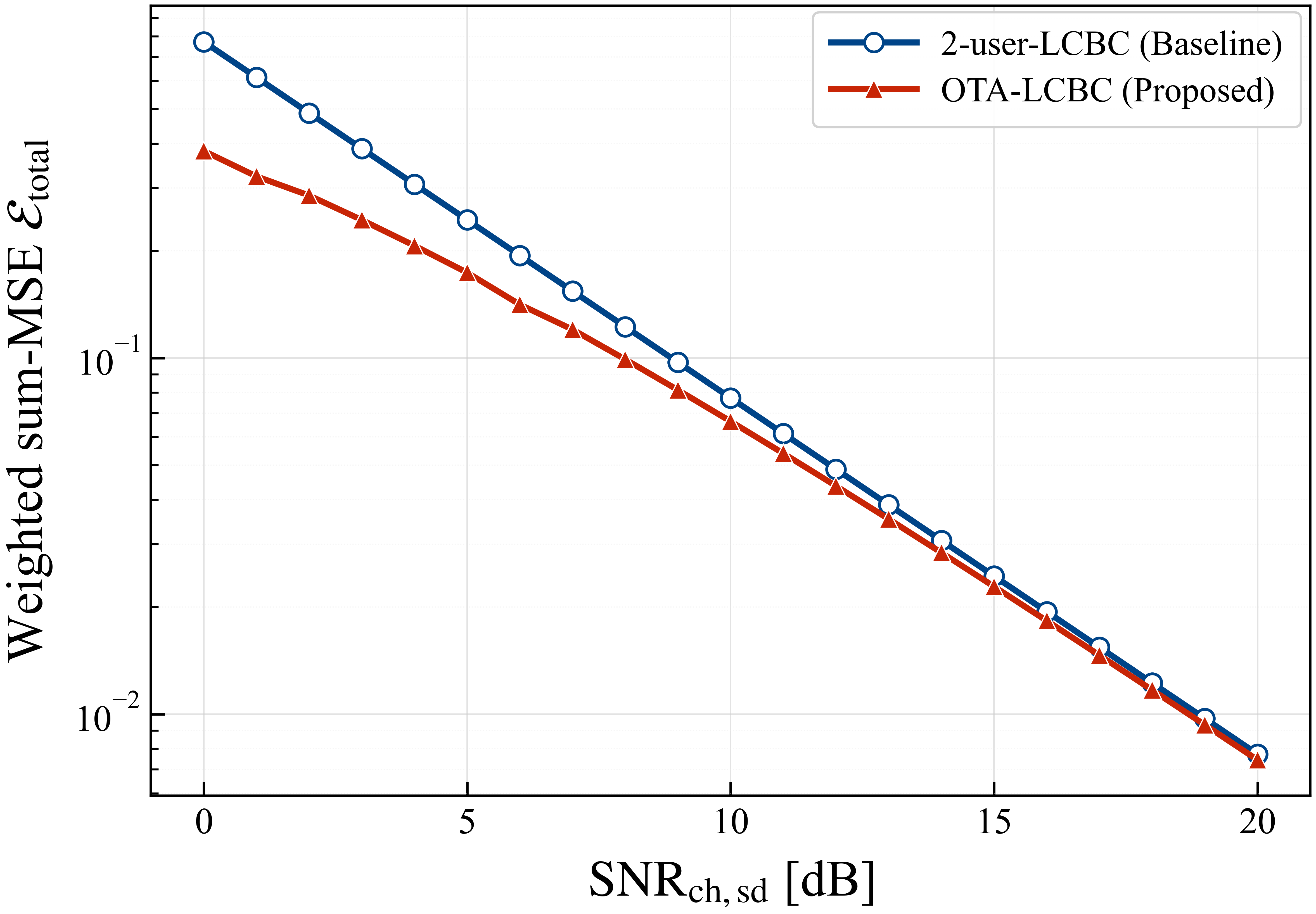}
  \caption{Two-user toy instance $(n,r)=(2,1)$ with
    $(d_1,d_1')=(1,1)$ and $(d_2,d_2')=(1,1)$:
    weighted sum-MSE $\mathcal E_{\mathrm{total}}$ versus
    $\gamma=\mathrm{SNR}_{\mathrm{ch,side}}$ (common channel and
    side-information SNR).
    Here $M=N_1=N_2=1$, $T=r=1$.}
  \label{fig:simple_2user}
\end{figure}

\emph{(b) Heterogeneous instance, precoder/decoder roles:}
We next consider a higher-dimensional heterogeneous configuration $n=7$, $r=4$, $(d_1,d_1')=(4,2)$, and $(d_2,d_2')=(3,2)$.
To pinpoint where the gain comes from, we compare three designs in Fig.~\ref{fig:complex_2user}:
(i) \emph{SJ-LCBC (precoder and decoder)}: the original SJ precoder and decoder applied to the noisy setting;
(ii) \emph{SJ precoder with OTA-optimized decoders}: the SJ precoder is kept fixed, while the decoders are re-optimized using the OTA decoder update;
(iii) \emph{WLCBC (AO precoder and decoder)}: both precoder and decoders are optimized by the proposed AO procedure. 

\begin{addedtext}Fig.~\ref{fig:complex_2user} shows that noise-aware decoder optimization provides substantial robustness gains in the low-SNR regime, but is not sufficient by itself. The remaining performance improvement, particularly at moderate and high SNRs, is primarily attributable to optimizing the precoder.
\end{addedtext}
At high SNR ($\gamma\gtrsim 30$~dB), the two curves that share the SJ
precoder nearly overlap, indicating that decoder mismatch becomes negligible
as noise vanishes; the remaining gap to WLCBC (roughly $8$ dB for $\mathcal E_{\mathrm{total}}\approx 10^{-2}$) reflects the benefit of a
noise-aware precoder that avoids ill-conditioned alignment and mitigates
noise enhancement.
\begin{figure}[!t]
  \centering
  \includegraphics[width=\columnwidth]{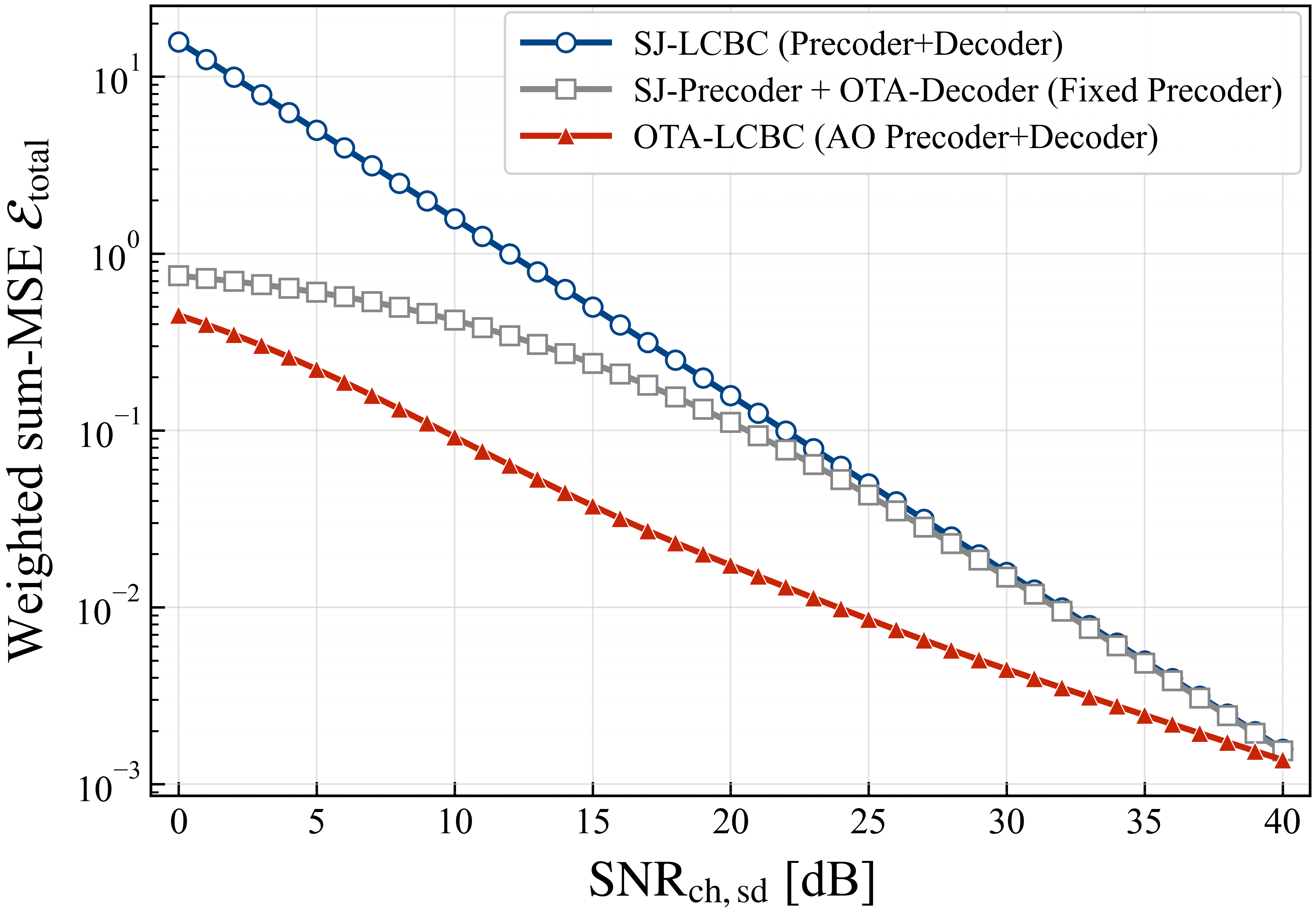}
  \caption{Heterogeneous two-user instance $(n,r)=(7,4)$ with $(d_1,d_1')=(4,2)$ and $(d_2,d_2')=(3,2)$:
  ablation comparing (i) SJ-LCBC with the Sun--Jafar precoder and decoder, (ii) the fixed Sun--Jafar precoder with OTA-optimized decoders,
  and (iii) full WLCBC with AO precoder and decoder, versus $\gamma=\mathrm{SNR}_{\mathrm{ch,side}}$ (common channel and side-information SNR). Here $M=N_1=N_2=1$, $T=r=4$.}
  \label{fig:complex_2user}
\end{figure}
\begin{figure*}[t]
  \centering
  \includegraphics[width=\linewidth]{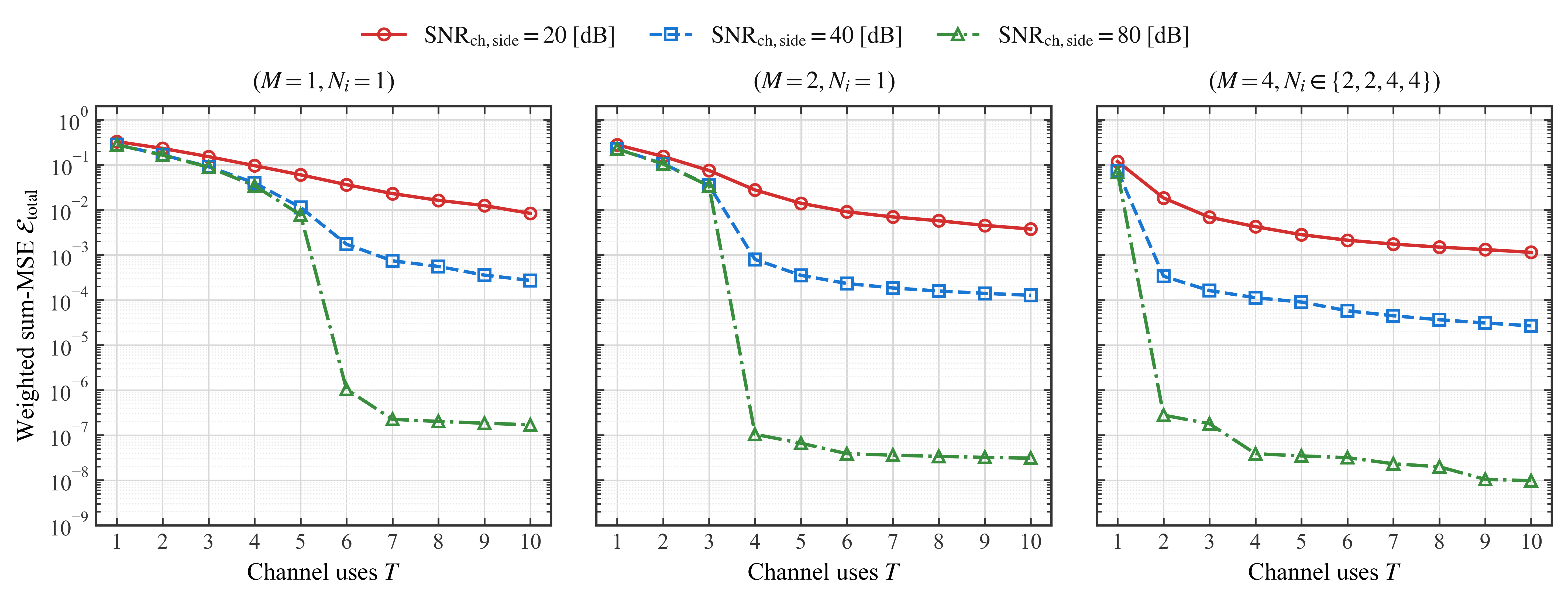}
  \caption{$K=4$-user MIMO Gaussian WLCBC, $n=24$ and $(d_i,d_i')=(4,12)$: total MSE $\mathcal E_{\mathrm{total}}$ versus channel uses $T$ for $\mathrm{SNR}_{\mathrm{ch}}\in\{20,40,80\}$~dB.
  The antenna configurations are shown above each subplot. Recall $r=TM$ and $m_i=TN_i$ in~\eqref{eq:stacking_def}, and $\|\mathbf E\|_F^2\le r$ in~\eqref{eq:power_constraint}.}
  \label{fig:exp2_mimo_Tsweep}
\end{figure*}
\begin{addedtext}
\subsection{Experiment 2: Channel-use-versus-distortion transition in Gaussian MIMO BC}
\label{subsec:exp2_mimo_Tsweep}
We evaluate WLCBC over i.i.d.\ real Gaussian MIMO broadcast channels
and investigate how the channel-use-versus-distortion tradeoff connects
to the communication-efficiency viewpoint of LCBC.
\end{addedtext}
\subsubsection{Gaussian MIMO setup and dimensionality scaling}
For each user $i$ and channel use $t$, the entries of
$\mathbf H_i^{(\mathrm p)}[t]\in\mathbb R^{N_i\times M}$ are drawn i.i.d.\ as
\[
[\mathbf H_i^{(\mathrm p)}[t]]_{a,b}
\sim
\mathcal N\!\left(0,\frac{1}{M}\right),
\]
independently across $(i,t)$.
The $1/M$ variance normalization yields unit average gain per receive antenna, i.e.,
\[
\mathbb E\!\left[
\mathbf H_i^{(\mathrm p)}[t]
\mathbf H_i^{(\mathrm p)}[t]^{\!\top}
\right]
=
\mathbf I_{N_i}.
\]

As in~\eqref{eq:stacking_def}, transmitting over $T$ channel uses induces the stacked dimensions
$r=TM$ and $m_i=TN_i$.
We impose the same block power constraint $\|\mathbf E\|_F^2\le r=TM$, which corresponds to unit average power per transmit dimension
(antenna--time slot).
Under the above channel normalization, the nominal received signal power per receive antenna satisfies
\[
\frac{1}{N_i}\,
\mathbb E\!\left[
\|\mathbf H_i^{(\mathrm p)}[t]\mathbf s[t]\|_2^2
\right]
=
\frac{1}{M}\,
\mathbb E\!\left[
\|\mathbf s[t]\|_2^2
\right]
\le 1,
\]
so increasing $M$ does not artificially increase SNR; it only increases the available spatial DoF in the stacked model.
Accordingly, we parameterize the channel noise variance directly by the target receive SNR (in dB) as
$\sigma_{\mathrm{ch}}^2 = 10^{-\gamma_c/10}$, so that $\gamma_c$ represents the nominal per-antenna receive SNR.

We fix a $K=4$-user instance with $n=24$ and $(d_i,d_i')=(4,12)$ for all $i\in\mathcal K$, and sweep $T\in\{1,\ldots,10\}$.
The side-information subspaces are generated as in Section~\ref{subsec:exp_setup}.
\begin{addedtext}The choice \(d_i'=12\) is made relative to the source dimension
\(n=24\).
Since \(\mathbf V_i\) and \(\mathbf V_i'\) are independently generated
random orthonormal bases, \(d_i'=n/2\) provides substantial but
incomplete side information: on average, half of the demand energy is
aligned with the side-information subspace, while the corresponding
observations remain noisy. In general, increasing \(d_i'\) reduces the
posterior residual uncertainty over more source directions and therefore
tends to lower the MSE and may shift the waterfall to a smaller \(T\).
\end{addedtext}

Throughout this experiment, the channel and side-information SNRs
are matched for each user. For
Fig.~\ref{fig:exp2_mimo_Tsweep}, we set
$\gamma_c=\gamma_s$ and sweep their common value over
$\{20,40,80\}$~dB.
Three antenna configurations are compared in Fig.~\ref{fig:exp2_mimo_Tsweep}:
(i) SISO $(M=1,\,N_i=1)$,
(ii) MISO $(M=2,\,N_i=1)$,
and (iii) heterogeneous MIMO $(M=4,\,N_i\in\{2,2,4,4\})$.

Because the per-channel-use/per-antenna power constraint is kept fixed, increasing $T$ increases both the number of available observations and the total energy available over the block.

\subsubsection{Results: Interference-limited to noise-limited transition}
Fig.~\ref{fig:exp2_mimo_Tsweep} exhibits a consistent
\emph{two-regime behavior}, most visible at high SNR.
For small $T$, the distortion is relatively insensitive to SNR: even
large reductions in noise yield only marginal MSE improvement.
This is the \emph{interference-limited regime}.
The stacked channel does not yet provide enough effective
spatio-temporal DoF to reliably resolve all users' demand components,
so the residual error is dominated by interference-limited geometry
rather than noise.
Once $T$ exceeds a configuration-dependent threshold $T^\star$, the
distortion undergoes a sharp ``waterfall'' and drops by multiple orders
of magnitude, after which the curves become strongly SNR-dependent,
consistent with a \emph{noise-limited regime}.

\begin{addedtext}
The observed waterfall can be interpreted through a simple dimension
count. The stacked transmit and receive dimensions are \(r=TM\) and
\(m_i=TN_i\), so the channel-provided source subspace available to user
\(i\) has dimension at most \(\min\{n,r,m_i\}\). In the present
random-subspace setting, \((d_i,d_i')=(4,12)\) and \(d_i+d_i'<n\), so
the demand generically contributes four dimensions outside the
side-information subspace. This yields the necessary per-user condition
\(\min\{TM,TN_i\}\ge4\), corresponding to \(T\ge4\) for SISO and MISO
and \(T\ge2\) for the bottleneck users in heterogeneous MIMO. These
counts are consistent with the observed MISO and MIMO waterfalls, while
the later SISO transition is consistent with additional multiuser
coupling through the common transmit subspace. Also, for a target distortion \(\delta\), the first simulated \(T\) at which the optimized MSE falls below \(\delta\) provides an empirical estimate of the required number of channel uses under the proposed WLCBC design.
\end{addedtext}
\begin{addedtext}
\subsubsection{Per-user distortion under heterogeneous SNRs}
\label{subsubsec:heterogeneous_user_snr}
We revisit its
MISO configuration with $K=4$, $n=24$, $(d_i,d_i')=(4,12)$,
$M=2$, $N_i=1$, and $T\in\{1,\ldots,10\}$.
Users $1$--$4$ have
$\mathrm{SNR}_{\mathrm{ch},i}
=\mathrm{SNR}_{\mathrm{side},i}
\in\{10,20,40,80\}$~dB, respectively, and are jointly optimized
with uniform weights and a common precoder.
\begin{figure}[!t]
  \centering
  \includegraphics[width=\columnwidth]{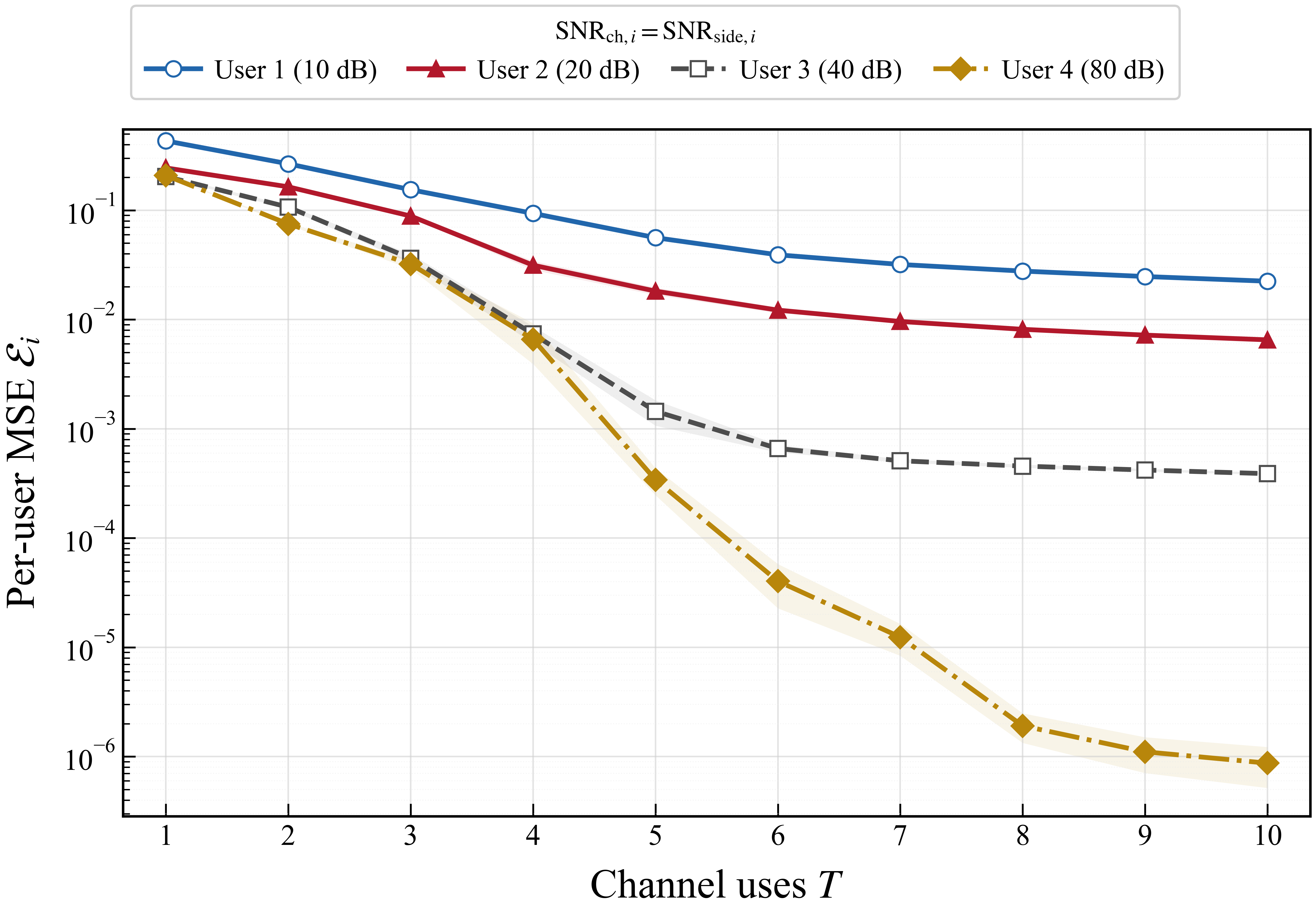}
  \caption{Per-user MSE versus the number of channel uses $T$ for
  $K=4$, $n=24$, $(d_i,d_i')=(4,12)$, $M=2$, and $N_i=1$.
  Users $1$--$4$ have
  $\mathrm{SNR}_{\mathrm{ch},i}
  =\mathrm{SNR}_{\mathrm{side},i}
  \in\{10,20,40,80\}$~dB, respectively.}
  \label{fig:heterogeneous_user_snr}
\end{figure}
Fig.~\ref{fig:heterogeneous_user_snr} reveals a regime-dependent
role of SNR.
For $T\le 4$, the $40$- and $80$-dB curves remain close despite
their $40$-dB SNR gap, and both undergo a waterfall near
$T=5$.
Their nearly common transition point suggests that, at high SNR, the
onset of the waterfall is governed primarily by the channel use bottleneck rather than by noise.
After this bottleneck is relieved, the two curves separate sharply
according to their respective noise levels.
Thus, SNR has limited leverage before sufficient transmission
dimensions become available, but becomes decisive afterward.
The $10$- and $20$-dB curves instead decrease more smoothly because
noise masks the sharp separation between the interference- and
noise-limited regimes.
Since all users share one precoder, these curves reflect a coupled
multiuser allocation rather than four independent point-to-point
links.

\subsubsection{AO convergence across SNRs}
\label{subsubsec:ao_convergence}
We next examine the iteration-wise behavior of AO for a fixed
heterogeneous-MIMO realization with $K=4$, $n=24$,
$(d_i,d_i')=(4,12)$, $M=4$,
$(N_1,N_2,N_3,N_4)=(2,2,4,4)$, and $T=2$.
The common channel and side-information SNR is $20$, $40$, or
$80$~dB.

Because the pointwise minimum of non-increasing sequences is also
non-increasing, the envelopes in Fig.~\ref{fig:ao_convergence}
inherit the same descent property.
The initial decrease is rapid: by iteration $10$, the objective is
reduced by factors of approximately $7$, $34$, and $66$ at $20$,
$40$, and $80$~dB, respectively.

The late-stage behavior, however, depends strongly on SNR.
From iteration $50$ to $250$, the $20$- and $40$-dB curves decrease
by only about $12\%$ and $25\%$, whereas the $80$-dB curve falls by
nearly three additional orders of magnitude.
This contrast is consistent with the lower-SNR cases becoming
noise dominated relatively early, while the much lower noise floor
at $80$~dB keeps small residual errors visible to subsequent
AO updates.
The terminal values,
$1.66\times10^{-2}$, $3.70\times10^{-4}$, and
$4.62\times10^{-8}$, also track the corresponding noise-variance
scales $10^{-2}$, $10^{-4}$, and $10^{-8}$.
Hence, the longer high-SNR tail reflects a lower observable error
floor rather than weaker descent.
\begin{figure}[!t]
  \centering
  \includegraphics[width=\columnwidth]{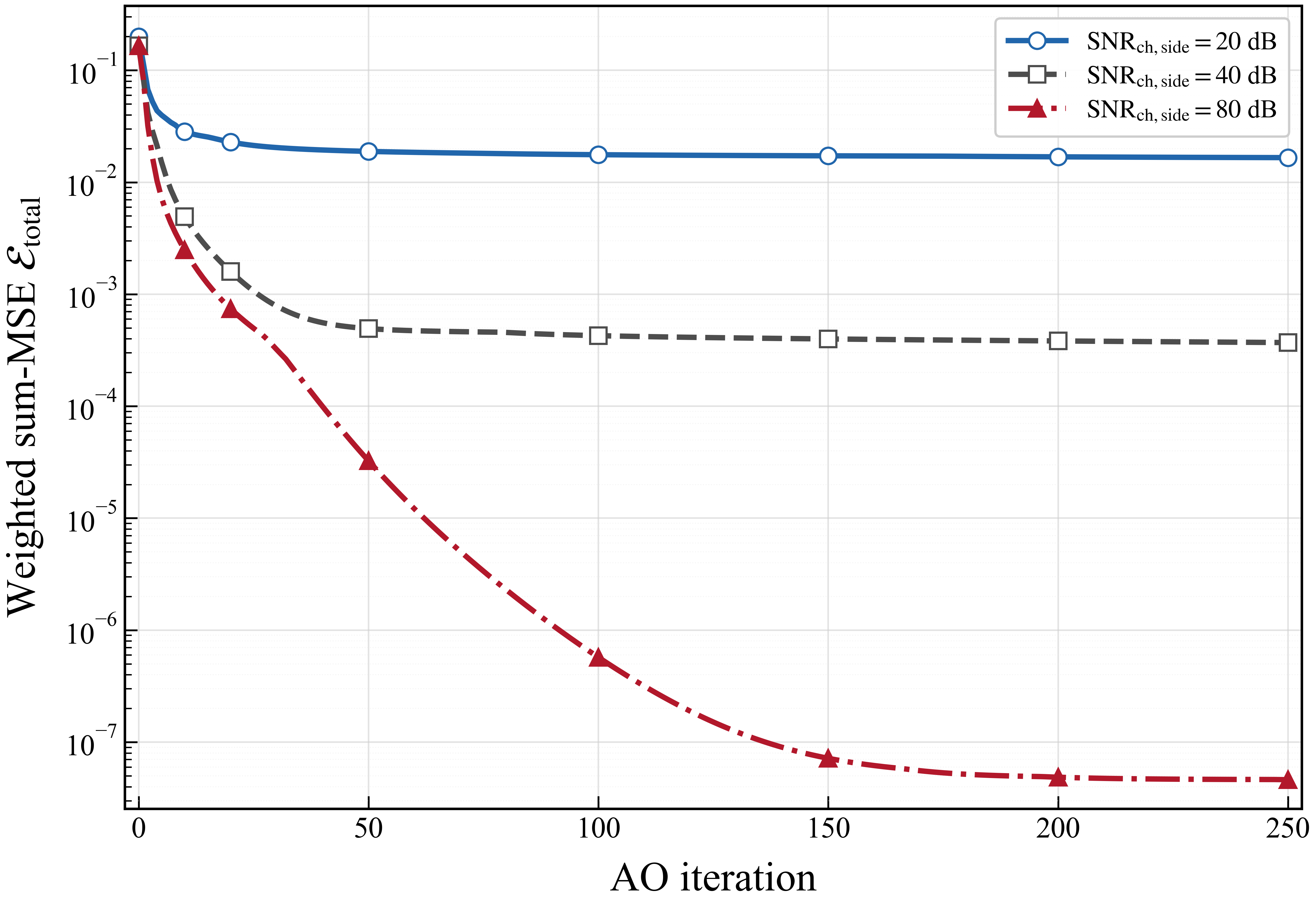}
  \caption{AO convergence (pointwise lower envelopes of the unclipped
  $\mathcal E_{\mathrm{total}}$) for $(K,n,M,T)=(4,24,4,2)$,
  $(d_i,d_i')=(4,12)$, $(N_i)_{i=1}^4=(2,2,4,4)$, and common
  channel/side-information SNRs of $20$, $40$, and $80$~dB.}
  \label{fig:ao_convergence}
\end{figure}

We further examine the sensitivity of AO to random initialization using
the same dimensions as in Fig.~\ref{fig:ao_convergence}. For each
$R\in\{1,5,20,50,100\}$, we report the minimum MSE
among the $R$ starts of the same random-start sequences.
As shown in Table~\ref{tab:ao_multistart}, the best objective values
stabilize after a moderate number of starts, with negligible improvement
from $R=50$ to $R=100$. Thus, AO exhibits limited sensitivity to the
random-start budget for this realization.

\begin{table}[!t]
\color{IEEEAddBlue}
\caption{Best MSE among $R$ random starts
on one fixed MIMO realization.}
\label{tab:ao_multistart}
\centering
\footnotesize
\setlength{\tabcolsep}{3.2pt}
\renewcommand{\arraystretch}{1.05}
\begin{tabular}{@{}rccc@{}}
\toprule
& \multicolumn{3}{c}{$\mathrm{SNR}_{\mathrm{ch,side}}$} \\
\cmidrule(lr){2-4}
$R$ & $20$ dB & $40$ dB & $80$ dB \\
\midrule
1
& $1.6267\times10^{-2}$
& $2.4701\times10^{-4}$
& $1.3741\times10^{-7}$ \\
5
& $1.6267\times10^{-2}$
& $2.2014\times10^{-4}$
& $6.4918\times10^{-8}$ \\
20
& $1.6081\times10^{-2}$
& $2.2014\times10^{-4}$
& $5.4543\times10^{-8}$ \\
50
& $1.6074\times10^{-2}$
& $2.0642\times10^{-4}$
& $5.4543\times10^{-8}$ \\
100
& $1.6074\times10^{-2}$
& $2.0642\times10^{-4}$
& $5.4543\times10^{-8}$ \\
\bottomrule
\end{tabular}
\end{table}
\subsection{Experiment 3: WLCBC versus a matched direct-output MIMO-BC baseline}
\label{subsec:exp3_same_task_mimo}
We next compare WLCBC with a matched \emph{same-task direct-output MIMO-BC}
baseline. Define
\[
\mathbf V_{\mathrm{tot}}
\triangleq
[\mathbf V_1,\ldots,\mathbf V_K],
\;
\mathbf w_{\mathrm{tot}}
\triangleq
[\mathbf w_1^\top,\ldots,\mathbf w_K^\top]^\top
=
\mathbf V_{\mathrm{tot}}^\top\mathbf x .
\]
Similar to traditional MIMO-BC schemes, the baseline considers only each user's 
requested output streams and then transmits their linear mixtures:
\begin{equation}
\mathbf s_{\mathrm{dir}}
=
\mathbf P\mathbf w_{\mathrm{tot}}
=
\mathbf P\mathbf V_{\mathrm{tot}}^\top\mathbf x,
\qquad
\mathbf E_{\mathrm{dir}}
=
\mathbf V_{\mathrm{tot}}\mathbf P^\top .
\label{eq:same_task_direct_output_2}
\end{equation}
Consequently, its feasible broadcast directions satisfy
\[
\operatorname{col}(\mathbf E_{\mathrm{dir}})
\subseteq
\underbrace{\Span(\mathbf V_{\mathrm{tot}})}_{
\substack{\text{direct-output baseline}\\\text{can use}}}
\subseteq
\underbrace{\mathbb R^n}_{\text{WLCBC can use}}.
\]
Thus, under the same transmit-dimension and power constraints, the
direct-output baseline's precoder is confined to the aggregate demand span $\Span(\mathbf V_{\mathrm{tot}})$, whereas
WLCBC can select precoding directions from the full source space $\mathbb R^n$. We also
equip the baseline with the same side-information-aware conditional LMMSE
decoder as WLCBC, so that the comparison focuses on this encoder-subspace
restriction rather than on decoder capability.

We use the same setting as in Experiment~2, except that
$\gamma_c=\gamma_s=40$~dB. Since $M=2$, the available transmit dimension is
$r=TM=2T$. The aggregate demand rank is
$
q\triangleq\rank(\mathbf V_{\mathrm{tot}})=16.
$
For each realization, the two schemes use identical demand and
side-information subspaces, channels, noise powers, and transmit-power
constraints.
\begin{figure}[htbp]
  \centering
  \includegraphics[width=\columnwidth]{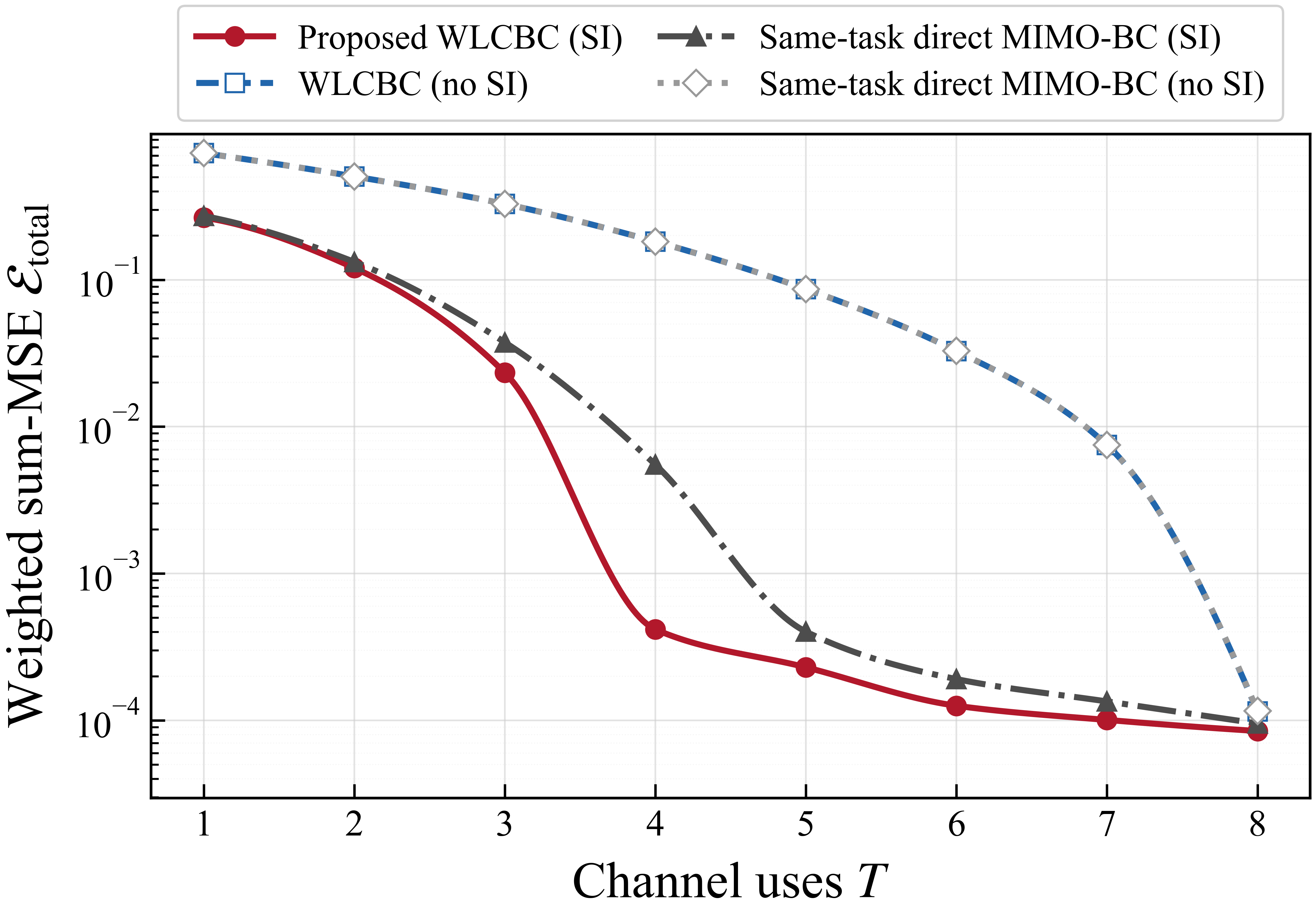}
  \caption{\added{Total MSE versus the number of channel uses
  $T$ in Experiment~3, with $K=4$, $n=24$, $(d_i,d_i')=(4,12)$,
  $M=2$, $N_i=1$, $\gamma_c=40$ dB, and $\gamma_s=40$ dB for the
  SI-enabled schemes. The matched same-task direct-output MIMO-BC
  baseline restricts $\operatorname{col}(\mathbf E)$ to
  $\Span(\mathbf V_{\mathrm{tot}})$, whereas WLCBC optimizes
  $\mathbf E$ over the full source space $\mathbb R^n$.
  ``SI'' denotes receiver side information.}}
  \label{fig:wlcbc_same_task_mimo}
\end{figure}
As shown in Fig.~\ref{fig:wlcbc_same_task_mimo}, without side information, WLCBC and the direct-output baseline perform nearly identically. In this case, any precoding directions outside $\Span(\mathbf V_{\mathrm{tot}})$ are uncorrelated with the requested outputs and merely waste transmit power. Thus, the optimal WLCBC precoder naturally lies within the aggregate demand span.

The presence of side information changes the geometry of useful transmitted
directions. For user $i$, the relevant conditional cross-covariance is
$
\operatorname{Cov}
\left(
\mathbf w_i,\mathbf s
\mid
\mathbf w_i'
\right)
=
\mathbf V_i^\top\bm\Pi_i\mathbf E,
\label{eq:conditional_cross_cov_exp3}
$
where $\bm\Pi_i=\operatorname{Cov}(\mathbf x\mid\mathbf w_i')$ is the
posterior residual covariance. Since $\bm\Pi_i$ does not necessarily preserve
$\Span(\mathbf V_{\mathrm{tot}})$, a direction
$\mathbf e\in\Span(\mathbf V_{\mathrm{tot}})^\perp$ may satisfy
$\mathbf V_i^\top\mathbf e=0$ but
$\mathbf V_i^\top\bm\Pi_i\mathbf e\neq0$. Such a direction is irrelevant to
the demanded outputs in the marginal source geometry, but can become
informative after conditioning on user $i$'s noisy side information. It
therefore acts as a side-information-complementary innovation direction.
Although the direct-output baseline exploits side information through its
LMMSE receivers and the associated encoder optimization, its transmitted
directions remain restricted to $\Span(\mathbf V_{\mathrm{tot}})$. In
contrast, WLCBC can additionally allocate transmit power to such conditional
innovation directions within the full source space. 

At $T=4$, incorporating side information reduces the MSE of the direct-output
baseline from approximately $1.82\times10^{-1}$ to $5.51\times10^{-3}$,
confirming that the baseline itself benefits substantially from this additional
information. However, WLCBC achieves a further order-of-magnitude improvement,
dropping the MSE to $4.17\times10^{-4}$---corresponding to an impressive
$92.4\%$ reduction over the side-information-aided baseline. Conversely, at $T=8$,
the available transmit dimension fully accommodates the aggregate demand rank
($r=q=16$). With the baseline's dimensional bottleneck largely eliminated,
the performance gap between the schemes naturally narrows.
\end{addedtext}
\section{Conclusion and Future Work}\label{sec:conclusion}

This paper studied WLCBC over real Gaussian MIMO broadcast channels with noisy receiver-side information under a centralized linear transceiver design formulation.
We formulated a WSMSE optimization problem under a unified stacked-channel power constraint and developed an efficient alternating-optimization framework.
\begin{addedtext}
The resulting AO procedure yields a monotonically non-increasing WSMSE objective sequence, and every accumulation point satisfies the KKT conditions of the joint design problem.
\end{addedtext}
For a fixed precoder, we established that the optimal LMMSE decoders admit an efficient implementation via a thin-SVD side-information reduction and a Schur-complement evaluation.
For fixed decoders, the precoder update is cast as a convex QCQP, for which we derived a semi-closed-form solution parameterized by a single optimal multiplier.
Numerical experiments demonstrated substantial robustness gains over classical two-user LCBC baselines in noisy regimes and revealed a clear channel-use-versus-distortion tradeoff across various antenna configurations.
\begin{addedtext}
The results further illustrated the effects of heterogeneous per-user SNRs and the convergence behavior of AO.
The matched-baseline comparison also showed that, in the presence of side information, WLCBC can outperform a direct-output MIMO-BC design by exploiting side-information-complementary directions in the full source space.
In particular, the observed high-SNR distortion waterfall provides an empirical indication of a spatio-temporal transmission-dimension bottleneck in the considered finite-dimensional linear WLCBC design.
\end{addedtext}

While the proposed framework is developed for real-valued Gaussian channels under perfect CSI, its underlying principles open several compelling avenues for future research.
\begin{addedtext}
The real-valued formulation admits a direct extension to proper complex Gaussian baseband models by replacing transposes with Hermitian transposes and using the corresponding complex covariance matrices.
Robust transceiver designs under imperfect or statistical CSI remain an important direction for future work.
\end{addedtext}
Furthermore, extending the current linear transceiver framework to accommodate nonlinear computational demands represents a promising frontier, particularly for functional delivery in distributed AI and machine learning applications.
Finally, investigating hybrid digital--analog schemes could yield additional robustness and hardware efficiency.

\section*{Acknowledgment}
The authors acknowledge the use of ChatGPT (OpenAI) to improve the readability of this manuscript, to assist in developing 
simulation codes, and to support the formatting of mathematical 
expressions. The authors have thoroughly reviewed and edited the text, independently verified all analytical derivations 
and simulation results, and assume full responsibility for the 
overall content and integrity of the final publication.
\bibliographystyle{IEEEtran}  
\bibliography{reference}
\begin{IEEEbiographynophoto}{Shuo Tan}
(Graduate Student Member, IEEE) received his B.E. degree in Computer Science (Excellence Plan) from Chongqing University, Chongqing, China, in 2025. He is currently pursuing the M.S./Ph.D. degrees with the Department of Electrical Engineering and Computer Science (EECS) at the University of California, Irvine (UCI), Irvine, CA, USA. He is a recipient of the UCI EECS Department Fellowship. His research interests include information theory and its applications.
\end{IEEEbiographynophoto}

\begin{IEEEbiographynophoto}{Syed A. Jafar}
(Fellow, IEEE) received his B. Tech. from IIT Delhi, India, in 1997, M.S. from Caltech, USA, in 1999, and Ph.D. from Stanford, USA, in 2003, all in Electrical Engineering. His industry experience includes positions at Lucent Bell Labs and Qualcomm. He is a Chancellor’s Professor in the Department of Electrical Engineering and Computer Science and the Henry Samueli Endowed Chair in Engineering at the University of California, Irvine (UCI), Irvine, CA, USA. His research interests include multiuser information theory, wireless communications and network coding.

Dr. Jafar is a recipient of the New York Academy of Sciences Blavatnik National Laureate in Physical Sciences and Engineering, the NSF CAREER Award, the ONR Young Investigator Award, the UCI Academic Senate Distinguished Mid-Career Faculty Award for Research, the School of Engineering Mid-Career Excellence in Research Award and the School of Engineering Maseeh Outstanding Research Award. His co-authored papers have received the IEEE Information Theory Society Paper Award, IEEE Communication Society and Information Theory Society Joint Paper Award, IEEE Communications Society Best Tutorial Paper Award, IEEE Communications Society Heinrich Hertz Award, IEEE Signal Processing Society Young Author Best Paper Award, IEEE Information Theory Society Jack Wolf ISIT Best Student Paper Award, and three IEEE GLOBECOM Best Paper Awards. Dr. Jafar received the UC Irvine EECS Professor of the Year award six times from the Engineering Students Council, a School of Engineering Teaching Excellence Award in 2012, and a Senior Career Innovation in Teaching Award in 2018. He was a University of Canterbury Erskine Fellow, an IEEE Communications Society Distinguished Lecturer, an IEEE Information Theory Society Distinguished Lecturer and a Thomson Reuters/Clarivate Analytics Highly Cited Researcher. He served as Associate Editor for IEEE TRANSACTIONS ON COMMUNICATIONS 2004-2009, for IEEE COMMUNICATIONS LETTERS 2008-2009 and for IEEE TRANSACTIONS ON INFORMATION THEORY 2009-2012.
\end{IEEEbiographynophoto}
\vfill
\end{document}